\documentclass[letterpaper,11pt]{article}

\usepackage{comment}	
\usepackage{amssymb, environ}
\usepackage{amsmath, amssymb, amsthm}
\usepackage[usenames, dvipsnames]{color}
\usepackage[normalem]{ulem} 
\usepackage{fullpage}
\usepackage[numbers, sort]{natbib}
\usepackage{tikz, pgfplots}
\usepackage{enumerate}
\usepackage{hyperref}
\usepackage{xspace,color}
\usepackage{graphicx}
\usepackage{caption}
\usepackage{subcaption}
\usepackage{enumitem,linegoal}
\usepackage[linesnumbered,ruled,vlined,algo2e]{algorithm2e}
\usepackage{ctable}					
\usepackage{mathtools}
\usepackage[flushleft]{threeparttable}
\usepackage{verbatim}

\usepackage{footnote}
\makesavenoteenv{tabular}
\makesavenoteenv{table}

\definecolor{crimsonglory}{rgb}{0.75, 0.0, 0.2}

\newcommand\Tstrut{\rule{0pt}{2.6ex}} 
\newcommand\Bstrut{\rule[-1.1ex]{0pt}{0pt}} 
\newcommand{\TBstrut}{\Tstrut\Bstrut} 

\newtheorem{theorem}{Theorem}[section]
\newtheorem{lemma}[theorem]{Lemma}
\newtheorem{corollary}[theorem]{Corollary}
\newtheorem{Definition}[theorem]{Definition}

\newcommand{\backursedit}{Backurs and Indyk}
\newcommand{\landauincremental}{Landau \textit{et al.}}

\newcommand{\Farhi}{Farhi \textit{et al.}}
\newcommand{\opt}{\textsf{opt}}
\newcommand{\nayebiquantum}{Nayebi \textit{et al.}}
\newcommand{\karloffmodel}{Karloff, Suri, and Vassilvitskii}
\newcommand{\lattanzifiltering}{Lattanzi \textit{et al.}}

\newcommand{\algorder}{\tildorder(n^{5/3}\poly(1/\epsilon))}
\newcommand{\thresholdestimation}{\textsf{threshold estimation}}
\newcommand{\lowerdistance}{l}
\newcommand{\upperdistance}{u}

\newcommand{\metric}{\mathcal{M}}
\newcommand{\errorvalue}{O(1/\epsilon)^{O(\log 1/\epsilon)}}
\newcommand{\tildorder}{\widetilde O}
\newcommand{\oracle}{\mathcal{O}}
\newcommand{\oraclep}{\mathcal{Q}}
\newcommand{\parity}{\mathsf{par}}
\newcommand{\ii}{\mathsf{I}}
\newcommand{\poly}{\mathsf{poly}}
\newcommand{\cor}{\mathsf{Cor}}
\newcommand{\metricspace}{\langle\metric, d\rangle}
\newcommand{\polylog}{\mathsf{polylog}}
\newcommand{\degreethreshold}[1]{{n_0}^{\tau_{#1}}}
\newcommand{\degreet}{\tau}

\newcommand{\edit}{\mathsf{edit}}
\newcommand{\lcs}{\mathsf{lcs}}
\newcommand{\dtw}{\mathsf{dtw}}
\newcommand{\fre}{\mathsf{fre}}
\newcommand{\errormetric}{\mathsf{e_{m}}}
\newcommand{\erroredit}{\mathsf{e_{e}}}
\newcommand{\timeedit}{\mathsf{t_{e}}}
\newcommand{\alg}{\mathcal{A}}
\newcommand{\mn}{\mathsf{min}}

\newcommand*\samethanks[1][\value{footnote}]{\footnotemark[#1]}

\newcounter{proccnt}

\newcommand{\konote}[1]{}

\title{Approximating Edit Distance in Truly Subquadratic Time: Quantum and MapReduce\footnote{Portions of this research were completed while the first, third, and fifth authors were visitors at the Simons Institute for the Theory of Computing.}
\footnote{A preliminary version of this paper was presented at SODA 2018.}
}

\author{
	Mahdi Boroujeni \thanks{Sharif University of Technology. Email: \texttt{safarnejad@ce.sharif.edu, ghodsi@sharif.edu}}
	\and Soheil Ehsani \thanks{University of Maryland. Email: \texttt{\{ehsani,hajiagha\}@cs.umd.edu, sseddigh@umd.edu}}
	\samethanks[6]
	\and Mohammad Ghodsi \samethanks[3]
	\thanks{Institute for Research in Fundamental Sciences (IPM).}
	\and MohammadTaghi HajiAghayi \samethanks[4]
	\thanks{Supported in part by NSF CAREER award CCF-1053605,  NSF BIGDATA grant IIS-1546108, NSF AF:Medium grant CCF-1161365,   
		DARPA GRAPHS/AFOSR grant FA9550-12-1-0423, and another DARPA SIMPLEX grant.}
	\and Saeed Seddighin \samethanks[4] \samethanks[6]
}

\begin{document}

\date{}

\maketitle

\begin{abstract}
The \textit{edit distance} between two strings is defined as the smallest number of \textit{insertions}, \textit{deletions}, and \textit{substitutions} that need to be made to transform one of the strings to another one. Approximating edit distance in subquadratic time is ``one of the biggest unsolved problems in the field of combinatorial pattern matching"~\cite{indyk2001algorithmic}. Our main result is a quantum constant approximation algorithm for computing the edit distance in truly subquadratic time. More precisely, we give an $O(n^{1.858})$ quantum algorithm that approximates the edit distance within a factor of $7$. We further extend this result to an $O(n^{1.781})$ quantum algorithm that approximates the edit distance within a larger constant factor.

Our solutions are based on a framework for approximating edit distance in parallel settings. This framework requires as black box an algorithm that computes the distances of several smaller strings all at once. For a quantum algorithm, we reduce the black box to \textit{metric estimation} and provide efficient algorithms for approximating it. We further show that this framework enables us to approximate edit distance in distributed settings. To this end, we provide a MapReduce algorithm to approximate edit distance within a factor of $3$, with sublinearly many machines and sublinear memory. Also, our algorithm runs in a logarithmic number of rounds.

\end{abstract}
\section{Introduction}\label{introduction}
The \textit{edit distance} (a.k.a \textit{Levenshtein distance}) is a well-known metric to measure the similarity of two strings. This metric has been extensively used in several fields such as computational biology, natural language processing, and information theory. The algorithmic aspect of the problem is even more fundamental; the problem of computing the edit distance is a textbook example for dynamic programming.

The edit distance between two strings is defined as the smallest number of \textit{insertions}, \textit{deletions}, and \textit{substitutions} that need to be made on one of the strings to transform it to another one. For two strings $s_1$ and $s_2$ with $n$ characters in total ($|s_1| + |s_2| = n$), a classic dynamic program finds the edit distance between them in time $O(n^2)$. The idea is to define auxiliary variables $d_{i,j}$'s which denote the edit distance between the first $i$ characters of $s_1$ and the first $j$ characters of $s_2$. Next, we iteratively determine the values of the auxiliary variables based on the following formula
\begin{equation*}
\small
d_{i,j} =
\begin{cases}
d_{i-1,j-1}, & \text{if }s_1[i] = s_2[j] \\
1+\min\{d_{i-1,j-1},d_{i,j-1},d_{i-1,j}\} & \text{if }s_1[i] \neq s_2[j].
\end{cases}
\end{equation*}
Despite the simplicity of the above solution, it has remained one of the most efficient algorithms from a theoretical standpoint to this day. Since the 1970s, several researchers aimed to improve the quadratic running time of the problem, however, thus far, the best-known algorithm runs in time $O(n^2/\log^2 n)$~\cite{masek1980faster}. The shortcoming of these studies is partly addressed by the work of \backursedit~\cite{backurs2015edit} wherein the authors show a truly subquadratic time algorithm is impossible to achieve unless a widely believed conjecture (\textsf{SETH}\footnote{The \textit{strong exponential time hypothesis} states that no algorithm can solve the satisfiability problem in time $2^{n(1-\epsilon)}$.}) fails.  

Unfortunately, the quadratic dependency of the running time on the size of the input makes it impossible to use such algorithms for large inputs in practice. For example, a human genome consists of almost three billion base pairs that need to be incorporated in similarity measurements. Therefore, several studies were focused on improving the running time of the algorithm by considering approximation solutions. A trivial $\sqrt{n}$ approximation algorithm follows from an $O(n + d^2)$ exact algorithm of \landauincremental~\cite{landau1998incremental} where $d$ is the edit distance between the two strings. Subsequent research improved this to $n^{3/7}$~\cite{bar2004approximating}, to $n^{1/3+o(1)}$~\cite{batu2006oblivious}, to $2^{\tildorder(\sqrt{\log n})}$~\cite{andoni2012approximating}, and the latest of which provides a polylogarithmic approximation guarantee in subquadratic time~\cite{andoni2010polylogarithmic}. Note that although the running times of these algorithms are almost linear, even if one favors the approximation factor over the running time, slowing down the algorithms to barely subquadratic doesn't yield an asymptotically better approximation guarantee. Despite persistent studies, finding a subquadratic algorithm with a constant approximation factor which is the ``holy grail" here is still open (see Section 6 of Indyk~\cite{indyk2001algorithmic}).

Quantum computation provides a strong framework to substantially improve the running time of many algorithmic problems. This includes a long list of problems from algebraic computational problems, to measuring graph properties, to string matching, to searching, to optimizing programs, etc.~\cite{ramesh2003string,farhi1999invariant,krovi2015quantum,jeffery2013nested,belovs2012learning,beals1997quantum,le2014improved,shor1994algorithms}. However, quantum  techniques can only be applied to limited structures. For instance, many classic problems such as sorting or even counting the number of 1's in a 0-1 array are still as time-consuming even with quantum computation. Indeed existing quantum techniques offer no immediate improvement to the running time of edit distance, neither to many classic DP-type problems such as finding the $\lcs$\enspace (longest common subsequence), $\dtw$\enspace (dynamic time wrapping) of two strings or determining the Fr\'{e}chet distance between two polylines. To the best of our knowledge, no exact or approximation algorithm is known for edit distance in subquadratic time via quantum computation. 

In this work, we provide a framework to approximate the edit distance between two strings within a constant factor. This framework requires as black box a procedure that takes several smaller strings as input and approximates their distances all at once. For quantum computers, we reduce this black box to finding the distances of a metric, namely \textit{metric estimation}. In this problem, we are given a metric space where any distance is available by a query from a distance oracle.
 We show that metric estimation cannot be approximated within a factor better than $3$ with a subquadratic number of quantum queries. On the contrary, we provide positive results for approximation factor $3$ and also larger constant factors. We show our bounds are tight up to constant factors by proving lower bounds on the query complexity of metric estimation.
Our metric estimation quantum algorithms are general tools and may find their applications in other distance-related problems as well.
Combining this black box with our framework yields subquadratic quantum algorithms for approximation edit distance within a constant factor.
Our work is similar in spirit to the work of Le Gall \cite{le2014improved} and D\"{u}rr \textit{et al.} \cite{doi:10.1137/050644719} where combinatorial techniques are used to obtain efficient quantum algorithms. We believe that our work opens an avenue to further investigation of edit distance in quantum setting and perhaps achieving near linear time quantum algorithm for edit distance.

As another application of our framework, we design a MapReduce algorithm for approximating edit distance within an approximation factor of $3$.
MapReduce is one of the most recent developments in the area of parallel computing. It has the benefits of both sequential and parallel computation. Many tech companies such as Google, Facebook, Amazon, and Yahoo designed MapReduce frameworks and have used them to implement fast algorithms to analyze their data.
In this paper, we focus on the well-known MapReduce theoretical framework initiated by Karloff, Suri, and Vassilvitskii \cite{karloff2010model} (and later further refined by Andoni, Nikolov, Onak, and Yaroslavtsev \cite{Andonimapreducestoc}).
Designing MapReduce algorithms for simulating sequential dynamic programs for important problems was recently initiated by Im, Moseley, and Sun \cite{mapreducedyn2017}.
They study DP-type problems with two key properties, monotonicity and decomposability. Their framework does not apply here since edit distance is neither monotone nor decomposable.
Our algorithm runs in a logarithmic number of rounds with a sublinear number of machines and sublinear memory of each machine. Moreover, the running time of each machine is subquadratic.

To the best of our knowledge, both our quantum algorithms and our MapReduce algorithm are first to improve upon the trivial $O(n^2)$ classic algorithm beyond subpolynomial factors for approximating edit distance\footnote{within a constant factor} in these settings. We believe that our framework can be useful to better understand edit distance in other models, such as the streaming and the semi-streaming models.

The closest works to our results are ~\cite{andoni2012approximating} and ~\cite{andoni2008smoothed}. In particular, they use a space embedding approach from \cite{Ostrovsky:2005:LDE:1060590.1060623} with dividing the string into blocks of smaller size, but our main observations and structural lemmas are completely different from their approach. We note that to the best of our knowledge, the ideas of our framework are novel and have not been used in any of the previous work. In ~\cite{apostolico1990efficient}, the authors give a parallel algorithm for determining the edit distance between two strings. Their algorithm uses $\tildorder(n^2)$ processors and a shared memory of  $O(n^2)$. Note that their algorithm cannot be used in MapReduce models, since the number of machines and memory of each machine in a MapReduce algorithm should be sublinear, and the number of rounds should be $O(\polylog(n))$ \cite{karloff2010model}. The major advantage of our MapReduce algorithm over the algorithm of ~\cite{apostolico1990efficient} is that both the number of machines and the memory of each machine is sublinear in our algorithm. Moreover, the number of rounds in our algorithm is $O(\log(n))$.

A similar approach is taken in the work of \nayebiquantum~\cite{nayebi2014quantum} wherein the authors study the computational complexity of APSP on quantum computers. They give an APSP algorithm for graph instances with small integer weights. They also give a fine-grained reduction from APSP to negative triangle via quantum computing.

\section{Our Results and Techniques} \label{section:ourResults}
In this section, we explain the ideas and techniques of our framework and show how we obtain a subquadratic algorithm for approximating the edit distance on quantum computers. The basis of our MapReduce algorithm is similar to what we explain here, though some details are modified to run the algorithm in a logarithmic number of MapReduce rounds. More details about the MapReduce algorithm can be found in Section \ref{mapreduce}.  Our quantum algorithm is based on several known techniques of quantum computing, algorithm design, and approximation algorithms. On the quantum side, we take advantage of \textit{Grover's search}~\cite{grover1996fast} and \textit{amplitude amplification}~\cite{brassard2002quantum} to improve the lookup time on an unordered set. On the algorithmic side, we benefit from classic algorithmic tools such as dynamic programming techniques, divide and conquer, and randomized techniques. In addition to this, we leverage \textit{the bootstrapping technique} to further improve the running time of our algorithm, by allowing the approximation guarantee to grow to larger constant numbers.

Recall that, the edit distance between two strings is defined as the smallest number of insertions, deletions, and substitutions, that one needs to perform on one of the strings to obtain the other one. For two strings $s_1$ and $s_2$, we denote their edit distance by $\edit(s_1,s_2)$. By definition, edit distance meets all of 
the \textit{identity of indiscernibles\footnote{$\edit(s_1,s_2) = 0 \Leftrightarrow s_1 = s_2$.}}, \textit{symmetry\footnote{$\edit(s_1,s_2) = \edit(s_2,s_1)$.}}, and \textit{triangle inequality\footnote{$\edit(s_1,s_2) + \edit(s_2,s_3) \geq \edit(s_1,s_3)$.}} properties, thus for any set of strings $\metric$, $\langle\metric,\edit\rangle$ forms a metric space\footnote{A set of points $\metric$ and a distance function $d$ form a metric space $\langle \metric, d\rangle$, if $d$ meets all of the aforementioned properties.}. Following this intuition, our algorithm is closely related to the study of the metric spaces.

In the following, we outline our algorithm in three steps. First, we define an auxiliary problem, namely \textit{metric estimation} and present efficient approximation algorithms for this problem accompanied by tight bounds on its quantum complexity. Roughly speaking, in this problem, we are given a metric space with $n$ points and oracle access to the distances, and the goal is to output an $n \times n$ matrix which is an estimate to the distances between the points. One may think of the oracle as an ordinary computer program, that we then convert to the corresponding quantum code and unitary operator using a quantum compiler \cite{Farhi:1998bz}. We give two approximation algorithms that solve the metric estimation problem with approximation factors $3+\epsilon$ and $\errormetric(\epsilon) = O(1/\epsilon)$ with $\tildorder(n^{5/3}\poly(1/\epsilon))$ and $\tildorder(n^{3/2+\epsilon}\poly(1/\epsilon))$ oracle queries, respectively. Notice that the running times of the algorithms are $O(n^2\poly(1/\epsilon))$, but the query complexities are subquadratic. This allows us to approximate metrics spaces with sublinear points for which answering an oracle query is time-consuming. We emphasize that our metric estimation results are general and can be used for any metric. In the second step, we show that any algorithm that solves the metric estimation problem within an approximation factor $\alpha$ can be used as a black box to obtain a $1+2\alpha+\epsilon$ approximation solution for edit distance. As we show, the reduction takes a subquadratic time and thus using our $3+\epsilon$ approximation algorithm for metric estimation, we obtain a $7+\epsilon$ approximation algorithm for edit distance. Finally, in Section \ref{bootstraping} we devise a bootstrapping technique to further improve the running time of the algorithm by taking a hit on the approximation guarantee. In what follows, we explain each of the steps in more details. Before we delve into the algorithm, we would like to note some comments.
\begin{itemize}	
	\item The only step of the algorithm where quantum computation plays a role is the first step where we discuss metric estimation. Nevertheless, everywhere we use the term algorithm, we mean a quantum algorithm unless otherwise is stated.
	
	\item In this section, we explain the abstract ideas and steps of the algorithm. Therefore, sometimes we do not provide formal proofs for some of the arguments that we make.
	The reader can find a detailed discussion of all statements and proofs in Sections \ref{metric} and \ref{editdistance}. 
	\item Everywhere we use the word \textit{operation}, we refer to insertion, deletion, or substitution.
\end{itemize}

\subsection{Metric Estimation}\label{contribution:metric}
As mentioned earlier, in the metric estimation problem, we are given a metric space $\langle \metric, d\rangle$ and an oracle $\oracle$ that reports $d(x,y)$ for two points $x$ and $y$ in an invocation. The goal of the problem is to estimate the distance matrix of the points with as few oracle calls as possible. Due to the impossibility results for exact or even solutions with small approximation factors for this problem (see the rest for more details), our aim is to find an approximation solution.
\begin{center}
	\noindent\framebox{\begin{minipage}{6.3in}
			\textsf{Metric Estimation}\\[0.25cm]
			\textsf{Input}: a metric space $\langle \metric,d \rangle$ with $n$ points where $\metric = \{p_1, p_2, \ldots, p_n\}$ and an oracle function $\oracle$ to access the distances. \\[0.25cm]
			\textsf{Guarantee}: all the distances are integer numbers in the interval $[l,u]$. We assume $u$ is $O(\poly(n))$.   \\[0.25cm]
			\textsf{An output (with approximation factor $\alpha > 1$)}: an $n \times n$ matrix $A$, where $d(p_i,p_j) \leq A[i][j] \leq \alpha d(p_i,p_j)$ holds for every $1 \leq i,j \leq n$.
		\end{minipage}}
\end{center}

Before we state the main ideas and results, we briefly explain two key tools that we borrow from previous work and use as black boxes in our algorithms. The first tool is the seminal work of Grover~\cite{grover1996fast} for making fast searches in an unordered database. Suppose we are given a function $f:[n] \rightarrow \{0,1\}$, where $[n] = \{1,2,3,\ldots,n\}$, and we wish to list up to $m$ distinct indices for which the value of the function is equal to $1$. We refer to this problem as \textit{element listing}.

\begin{center}
	\noindent\framebox{\begin{minipage}{6.3in}
			\textsf{Element Listing}\\[0.25cm]
			\textsf{Input}: integers $n$ and $0 \leq m \leq n$, and access to an oracle that upon receiving an integer $i$, reports the value of $f(i)$. $f$ is defined over $[n]$ and maps each index to either $0$ or $1$.  \\[0.25cm]
			\textsf{Output}: a list of up to $m$ indices for which the value of $f$ is equal to $1$. If the total number of such indices is not more than $m$, the output should contain all of them.
		\end{minipage}}
	\end{center}
The pioneering work of Grover~\cite{grover1996fast} implies that the element listing problem can be solved with only $O(\sqrt{nm})$ oracle calls via quantum computation. We subsequently make use of this algorithm in this section. 

\begin{theorem}[proven in~\cite{boyer1996tight}]\label{grover}
The listing problem can be solved with $O(\sqrt{nm})$ oracle queries via quantum computation.
\end{theorem}

The second quantum technique that we use in this paper is a tool for
 proving lower bounds on the quantum complexity of the problems. Let $f:[n] \rightarrow \{-1,1\}$ be a function defined over the numbers $1,2,\ldots,n$ 
 that maps each index to either $-1$ or $1$ and $\parity(f) = \prod_{i \in [n]}f(i)$.
  In the parity problem, we are given oracle access to $f$ and the goal is to determine $\parity(f)$ with as few oracle calls as possible.

\begin{center}
	\noindent\framebox{\begin{minipage}{6.3in}
			\textsf{Parity}\\[0.25cm]
			\textsf{Input}: an integer $n$, and access to an oracle $\oracle$ that upon receiving an integer $i$ reports the value of $f(i)$. $f$ is defined over $[n]$ and maps each index to either $-1$ or $1$.  \\[0.25cm]
			\textsf{Output}: $\parity(f) = \prod_{i \in [n]}f(i)$.
		\end{minipage}}
	\end{center}
Of course, if the numbers of $-1$'s or $1$'s are substantially smaller than $n$ ($o(n)$), one can use Grover's search to list all of such indices and compute the parity with fewer than $\Omega(n)$ oracle calls. However, if this is not the case for either $-1$ or $1$, such an approach fails. The seminal work of \Farhi~\cite{Farhi:1998bz}, showed that at least $\Omega(n)$ queries are necessary for solving the parity problem and thus quantum computation offers no speedup in this case.

\begin{theorem}[proven in~\cite{Farhi:1998bz}]
	The parity problem cannot be solved with fewer than $\Omega(n)$ queries with quantum computation.
\end{theorem}
 
Based on the result of \Farhi~\cite{Farhi:1998bz}, we begin with showing an impossibility result. Our first result for metric estimation is a hardness of approximation for factors smaller than $3$ using a subquadratic number of queries. More precisely, in Section \ref{metric}, we show that any quantum algorithm that approximates  metric estimation  within a factor smaller than $3$, needs to make at least $\Omega(n^2)$ oracle queries.

\vspace{0.2cm}
{\noindent \textbf{Theorem} \ref{hardness} [restated]. \textit{Any quantum algorithm for solving the metric estimation problem with an approximation factor smaller than 3 needs to make at least $\Omega(n^2)$ oracle calls.\\}}

The idea is to show a reduction from parity to metric estimation. Suppose we are given an instance $\ii$ of the parity problem. Roughly speaking, we construct an instance $\cor(\ii)$ of the metric estimation and prove that $\cor(\ii)$ has a valid metric as input. Next, we show that any algorithm that approximates metric estimation within a factor smaller than $3$ with $o(n^2)$ queries can be turned into a quantum algorithm for solving parity with $o(n)$ queries which is impossible due to \Farhi~\cite{Farhi:1998bz}.

Despite this hardness of approximation for factors better than $3$, we show the problem is significantly more tractable when we allow the approximation guarantee to be slightly more than $3$. In Section \ref{metric}, we show that for any $\epsilon > 0$, a $3+\epsilon$ approximation of metric estimation is possible via $\tildorder(n^{5/3}\poly(1/\epsilon))$ queries.

\vspace{0.2cm}
{\noindent \textbf{Theorem} \ref{thm:metric1} [restated]. \textit{For any $\epsilon > 0$, there exists a quantum algorithm that solves metric estimation with $\tildorder(n^{5/3}\poly(1/\epsilon))$ queries within an approximation factor of $3+\epsilon$. Moreover, the running time of the algorithm is $\tildorder(n^2\poly(1/\epsilon))$.\\}}

Our first take on the solution is to discretize the problem at the expense of imposing an additional $1 + \epsilon$  factor to our guarantee. Notice that all of the distances of the metric lie in the interval $[l,u]$. Therefore, one can divide the distances into $\log_{1+\epsilon/3} (u/l) = \tildorder(\poly(1/\epsilon))$ disjoint intervals where the distances within each interval differ in at most a multiplicative factor of $1+\epsilon/3$. For every interval $[x,(1+\epsilon/3)x]$ we can set a threshold $t = (1+\epsilon/3)x$ and find all pairs within a distance of at most $t$ with an approximation factor of $3$. Then, based on all  these solutions, one can find a $3+\epsilon$ approximation distance for every pair of the points.

Now the problem boils down to the following: given a threshold $t$, find all pairs $(p_i,p_j)$ such that $d(p_i,p_j) \leq t$. Of course, an exact solution for this problem is hopeless due to our impossibility result. Therefore we allow some false positive in our solution as well. More precisely, we restrict our solution to contain all pairs $(p_i,p_j)$ such that $(p_i,p_j) \leq d$, but additional pairs are also allowed to appear, if $(p_i,p_j) \leq 3d$. It is easy to show that any solution that solves the above problem via $\tildorder(n^{5/3}\poly(1/\epsilon))$ queries, yields a $3+\epsilon$ approximation factor algorithm for metric estimation that uses at most $\tildorder(n^{5/3}\poly(1/\epsilon))$ oracle calls.

In what follows, we describe the ideas to solve the problem for a fixed threshold $t$. The algorithm is explained in details in Section \ref{metric}, therefore, here, we just mention the tools and techniques. For convenience, we construct a graph $G$ with $n$ nodes, and correspond every point $p_i$ of the metric to a vertex $v_i$ of the graph. For a pair of points $(p_i,p_j)$, we add an undirected edge $(v_i,v_j)$ to the graph, if $d(p_i,p_j) \leq t$. Notice that the oracle function $\oracle$, provides us the exact value of $d(p_i,p_j)$ for any $p_i$ and $p_j$, therefore we can examine whether an edge exists between two vertices $v_i, v_j$ with a single oracle call. Recall that, Grover's search allows us to find as many as $m$ elements with value $1$ of a function of size $n$ via $O(\sqrt{nm})$ oracle calls. Therefore, if the number of the edges of the graph is $O(n^{4/3})$, we can use Grover's search (Theorem \ref{grover}) to list all of the edges with $O(\sqrt{n^2\cdot n^{4/3}}) = O( n^{5/3})$ queries and solve the problem. Therefore, the non-trivial part of the problem is the case where the graph is dense. In this case, the average degree of the vertices is at least $\Omega(n^{1/3})$. Now, suppose we select a vertex $v_i$ whose degree is at least $n^{1/3}$, and with $n-1$ query calls, find the distances of its corresponding point $p_i$ from all other points of the metric. Let set $D^t$, be the set of all points that have a distance of at most $t$ from $p_i$ and $D^{2t}$ be the of points with a distance of at most $2t$ from $p_i$. Trivially, $D^t \subseteq D^{2t}$. Due to the triangle inequality, all of the edges incident to the vertices corresponding to set $D^t$ are from the vertices corresponding to $D^{2t}$. Moreover, the distances of all points of $D^t$ from points of $D^{2t}$ are bounded by $3t$. Therefore, one can report all such pairs in the solution and proceed by removing $D^t$ from the graph (however, some vertices of $D^{2t}$ remain in the graph). Thus, all that remains is to solve the problem for an instance with at most $n-n^{1/3}$ nodes recursively. Since we make at most $O(n)$ query calls for every $n^{1/3}$ vertices (an amortized of $n^{2/3}$ per vertex), the total number of queries is $O(n^{5/3})$. More details about this can be found in Section \ref{metric}.

In addition to Theorem \ref{thm:metric1}, we show in Section \ref{metric} that with a deeper analysis, one can use the same ideas to further improve the query complexity to $\tildorder(n^{3/2+\epsilon}\poly(1/\epsilon))$ by allowing the approximation guarantee to grow up to $\errormetric(\epsilon) = O(1/\epsilon)$. 

\vspace{0.2cm}
{\noindent \textbf{Theorem} \ref{thm:metric151} [restated]. \textit{For any $\epsilon > 0$, there exists a quantum algorithm that solves metric estimation with $\tildorder(n^{3/2+\epsilon}\poly(1/\epsilon))$ queries within an approximation factor of $\errormetric(\epsilon) = O(1/\epsilon)$.  Moreover, the running time of the algorithm is $\tildorder(n^2\poly(1/\epsilon))$.\\}}

You can find a summary of the results explained in this section in Table \ref{table}.

\begin{table}[t]\centering\small
	\caption{Quality of the approximation algorithms for metric estimation}
	\begin{tabular}{|l|c|c|c|c|c|}
		\hline
		\TBstrut
		 \begin{tabular}{@{}l@{}}Approx. \\ factor\end{tabular} & $\alpha < 3$ & $\alpha = 3+\epsilon$& $\alpha = \errormetric(\epsilon)$ & $\alpha= $ any constant\\
		\hline
		\Tstrut
		Number of& $\Omega(n^2)$   & $\tildorder(n^{5/3}\poly(1/\epsilon))$&  $\tildorder(n^{3/2+\epsilon}\poly(1/\epsilon))$ & $\Omega(n^{3/2})$ \\
\Bstrut
		queries & (Theorem \ref{hardness}) & (Theorem \ref{thm:metric1}) & (Theorem \ref{thm:metric151}) & (Theorem \ref{thm:lowerquery})\\
		\hline
	\end{tabular}
	\label{table}
\end{table}

\subsection{Approximating Edit Distance within a Factor $7+\epsilon$}\label{ghabli}
In the second step, we provide an algorithm to approximate the edit distance between two strings in subquadratic time, based on a reduction to metric estimation. Our approach here is twofold. Suppose we are given a \textit{guess} $d$, on the actual edit distance between the strings, and we want to find an approximation proof to the guess. More precisely, we wish to find out whether $d$ is smaller than the actual distance of the strings, or report a transformation of the strings with at most $\alpha d$ operations\footnote{insertion, deletion, or substitution} where $\alpha$ is given as an approximation factor. If $d$ is substantially smaller than $n$, then the $O(n + d^2)$ exact algorithm of \landauincremental~\cite{landau1998incremental} solves the problem in subquadratic time. Therefore, the only hard instances of the problem are when $d$ is asymptotically close to $n$. Therefore, we define a subtask of the edit distance problem, in which we are given two strings $s_1$ and $s_2$ and guaranteed that the edit distance between the strings is at most $\delta (|s_1|+|s_2|)$ where $\delta$ is not too small. The goal is to find a transformation of the strings with at most $(\delta\cdot \alpha) (|s_1|+|s_2|)$ operations, where $\alpha$ is the approximation factor of the algorithm. We refer to this subtask of edit distance as \textit{the $\delta$-bounded edit distance} problem.

\begin{center}
	\noindent\framebox{\begin{minipage}{6.3in}
			\textsf{$\delta$-bounded edit distance}\\[0.25cm]
			\textsf{Input}: two strings $s_1$ and $s_2$, and a real number $0 \leq \delta \leq 1$.  \\[0.25cm]
			\textsf{Guarantee}: $\edit(s_1,s_2) \leq \delta(|s_1|+|s_2|)$.  \\[0.25cm]
			\textsf{Output (with an approximation factor $\alpha>1$)}: a sequence of operations with size at most $(\delta\cdot \alpha) (|s_1|+|s_2|)$ that transforms $s_1$ into $s_2$.
		\end{minipage}}
\end{center}

We combine a divide and conquer technique with dynamic programming in order to approximate $\delta$-bounded edit distance. In addition to this, we subsequently make use of the quantum techniques mentioned earlier in our solution.
Recall that the total number of characters in the input is equal to $n$, i.e., $|s_1|+|s_2| = n$. For clarity, we define two parameters $0 < \beta < 1$ and $\gamma > 1$. $\gamma$ is an integer number but $\beta$ is a real number between 0 and 1. We use $\beta$ and $\gamma$ as two parameters of our algorithm, and after the analysis, we show which values for $\beta$ and $\gamma$ give us the best guarantee.

We begin by defining the notion of a \textit{window} and construct a set of windows for each string. Let $l = \lfloor n^{1-\beta}\rfloor$ be the \textit{window size} and define a window of $s_1$, as a string of length $l$ over the characters of $s_1$. Moreover, define $g = \lfloor l/\gamma\rfloor = O(n^{1-\beta}/\gamma)$ as \textit{the gap size} and construct a collection $W_1$ of windows for $s_1$ as follows: for every $0 \leq i \leq \lfloor \frac{|s_1|-l}{g} \rfloor$, put a window $[ig+1,ig+l]$ (i.e., a window from index $ig+1$ to index $ig+l$ of $s_1$) in $W_1$. In other words, $W_1$ contains tentatively $\gamma(|s_1|/l) = O(\gamma n^\beta)$ windows of length $l$ where the gap between the neighboring windows is equal to $g$. Figure \ref{figs:windows} illustrates how the windows of $W_1$ span over the characters of $s_1$. Notice that some of the windows overlap. 
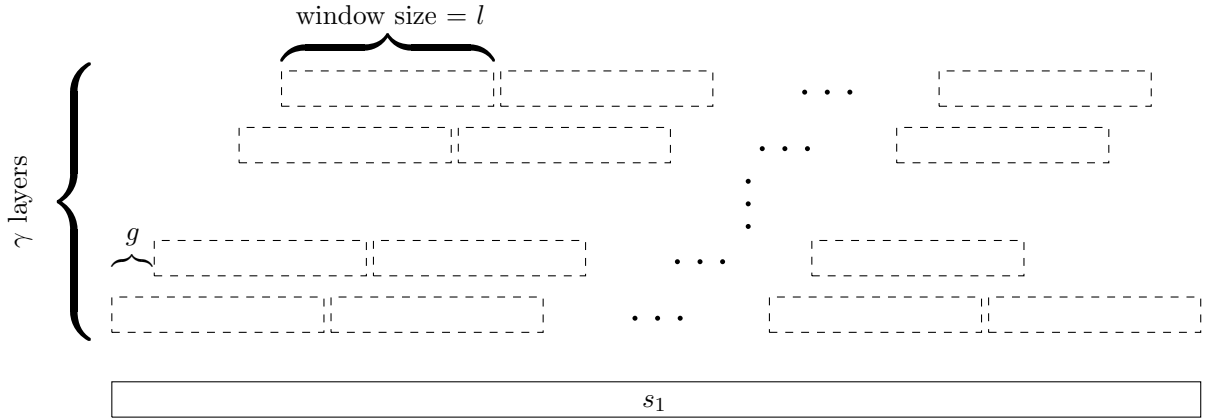
\begin{figure*}[h!]

\begin{center}
\begin{tikzpicture}[scale=0.94, transform shape]

\draw (0,0) -- (15.4,0) -- (15.4,0.5) -- (0,0.5) -- (0,0);
\node[text width=3cm] at (9,0.2) 
{$s_1$};

\node[text width=3cm, rotate = 90] at (-0.5,2.6) 
{    \Huge$\overbrace{\hspace{3.9cm}}$   };
\node[text width=3cm, rotate = 90] at (-1.3,3.9) 
{  $\gamma$ layers  };

\node[text width=3cm] at (3.9,5.2) 
{    \huge$\overbrace{\hspace{3.0cm}}$   };

\node[text width=3cm] at (4.1,5.7) 
{    window size = $l$   };

\node[text width=2.8cm] at (1.4,2.2) 
{    \footnotesize $\overbrace{\hspace{0.3cm}}$   };

\node[text width=3cm] at (1.7,2.55) 
{    $g$   };

\node[text width=3cm, rotate = 90] at (9,4.1) 
{  \huge  $\ldots$   };

\draw[dashed] (0,1.2) -- (3,1.2) -- (3,1.7) -- (0,1.7) -- (0,1.2);
\draw[dashed] (3.1,1.2) -- (6.1,1.2) -- (6.1,1.7) -- (3.1,1.7) -- (3.1,1.2);
\node[text width=3cm] at (8.8,1.4) 
{   \huge $\ldots$   };

\draw[dashed] (9.3,1.2) -- (12.3,1.2) -- (12.3,1.7) -- (9.3,1.7) -- (9.3,1.2);
\draw[dashed] (12.4,1.2) -- (15.4,1.2) -- (15.4,1.7) -- (12.4,1.7) -- (12.4,1.2);

\draw[dashed] (0.6,2) -- (3.6,2) -- (3.6,2.5) -- (0.6,2.5) -- (0.6,2);
\draw[dashed] (3.7,2) -- (6.7,2) -- (6.7,2.5) -- (3.7,2.5) -- (3.7,2);
\node[text width=3cm] at (9.4,2.2) 
{   \huge $\ldots$   };
\draw[dashed] (9.9,2) -- (12.9,2) -- (12.9,2.5) -- (9.9,2.5) -- (9.9,2);


\draw[dashed] (1.8,3.6) -- (4.8,3.6) -- (4.8,4.1) -- (1.8,4.1) -- (1.8,3.6);
\draw[dashed] (4.9,3.6) -- (7.9,3.6) -- (7.9,4.1) -- (4.9,4.1) -- (4.9,3.6);
\node[text width=3cm] at (10.6,3.8) 
{   \huge $\ldots$   };

\draw[dashed] (11.1,3.6) -- (14.1,3.6) -- (14.1,4.1) -- (11.1,4.1) -- (11.1,3.6);

\draw[dashed] (2.4,4.4) -- (5.4,4.4) -- (5.4,4.9) -- (2.4,4.9) -- (2.4,4.4);
\draw[dashed] (5.5,4.4) -- (8.5,4.4) -- (8.5,4.9) -- (5.5,4.9) -- (5.5,4.4);
\node[text width=3cm] at (11.2,4.6) 
{   \huge $\ldots$   };
\draw[dashed] (11.7,4.4) -- (14.7,4.4) -- (14.7,4.9) -- (11.7,4.9) -- (11.7,4.4);


\end{tikzpicture}
\end{center}
\caption{$s_1$ is shown with a solid rectangle and windows of $W_1$ are depicted via dashed rectangles.}\label{figs:windows}
\end{figure*}

Similar to this, we construct a collection $W_2$ of windows for $s_2$, using the same parameters $l$ and $g$. We define a \textit{transformation} of $s_1$ into $s_2$, as a sequence of insertions, deletions, and substitutions that turns $s_1$ into $s_2$. After a transformation of $s_1$ into $s_2$, we call a character of $s_2$ \textit{old} if it is either substituted by a character of $s_1$, or remained intact during the transformation. In other words, if a character is not inserted during a transformation, it is called old. Based on this, we define the notion of \textit{a window-compatible transformation} as follows:
\begin{Definition}
Let  $S = \langle w_1,w_2,\ldots,w_{k}\rangle$ and $S' = \langle w'_1,w'_2,\ldots,w'_{k}\rangle$ be two sequences of size $k$ of non-overlapping windows from $W_1$ and $W_2$, respectively. We call a transformation of $s_1$ into $s_2$ window-compatible with respect to $S$ and $S'$, if (i) all old characters of $s_2$ are in the windows of $S'$ and (ii) every old character of $s_2$ which is in a window $w'_i$, was placed in window $w_i$ of $s_1$ prior to the transformation. We call a transformation window-compatible, if it is window-compatible with respect to at least a pair of sequences of non-overlapping windows \textbf{from} $W_1$ and $W_2$, respectively.
\end{Definition}
Intuitively, a window-compatible transformation with respect to two sequences of windows $S$ and $S'$ does not allow the characters to move in between the windows; if a character is initially placed in a window $w_i$, it should either be deleted or placed in window $w'_i$ of $s_2$ and vice versa. We emphasize that in order for a transformation to be window-compatible, the corresponding windows should be selected from $W_1$ and $W_2$, respectively. A few examples of window-compatible and window-incompatible transformations are illustrated in Figure \ref{figs:examples}.

\begin{figure*}[h!]
    \centering
    \begin{subfigure}[t]{0.32\textwidth}
        \includegraphics[width=\textwidth]{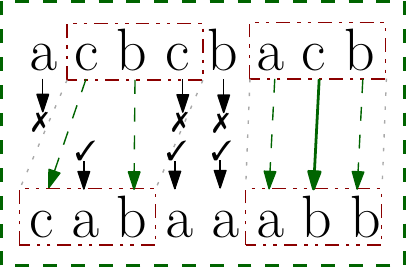}
        \caption{An example of a window-compatible transformation.\vspace{0.3cm}}
        \label{fig:e1}
    \end{subfigure}
    \qquad 
    \begin{subfigure}[t]{0.32\textwidth}
        \includegraphics[width=\textwidth]{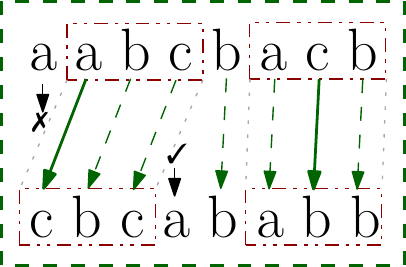}
        \caption{The transformation is not window-compatible since character 5 of the second string is old, but doesn't lie in any windows.\vspace{0.3cm}}
        \label{fig:e2}
    \end{subfigure}
    \qquad 
    \begin{subfigure}[t]{0.32\textwidth}
        \includegraphics[width=\textwidth]{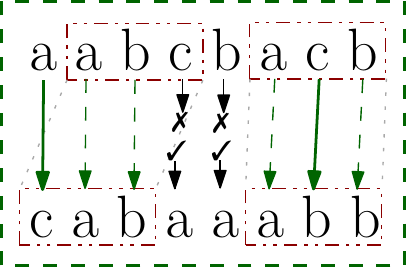}
        \caption{The transformation is not window-compatible since character 1 of the second string is old, but prior to the transformation, it was not placed in any windows.}
        \label{fig:e3}
    \end{subfigure}
    \qquad 
    \begin{subfigure}[t]{0.32\textwidth}
        \includegraphics[width=\textwidth]{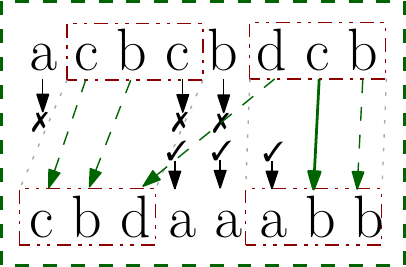}
        \caption{The transformation is not window-compatible since character 3 of the second string is old, but prior to the transformation, it was not placed in the corresponding window.}
        \label{fig:e4}
    \end{subfigure}
    \caption{Figures \ref{fig:e1}, \ref{fig:e2}, \ref{fig:e3}, and \ref{fig:e4} show a few examples of window-compatible and window-incompatible transformations. Solid arrows show substitutions, dashed arrows show the characters that remain in the string, and other characters are either inserted or deleted.}\label{figs:examples}
\end{figure*}

As we show in the following, window-compatible transformations are well-structured. In fact, we show in Section \ref{editdistance} that if the edit distances of the windows are accessible in time $O(1)$, a dynamic program can find an optimal\footnote{a transformation with the smallest number of operations.} window-compatible transformation of $s_1$ into $s_2$ in time $O(n+|W_1||W_2|)$.

\vspace{0.2cm}
{\noindent \textbf{Lemma} \ref{dp} [restated]. \textit{Given a matrix of edit distances between the substrings corresponding to every pair of windows of $W_1$ and $W_2$, one can compute an optimal window-compatible transformation of $s_1$ into $s_2$ in time $O(n+|W_1||W_2|)$.\\}}

Lemma \ref{dp} shows that window-compatible transformations are easy to find. It also follows from Lemma \ref{dp} that any $\alpha$ approximation matrix for the edit distances of the windows suffices to find an approximately optimal window-compatible transformation (with the same approximation factor) in time $O(n+|W_1||W_2|)$. This makes the connection of edit distance and metric estimation more clear. 

We complement this observation by a structural proof. In Section \ref{editdistance}, we show that the length of the shortest window-compatible transformation of $s_1$ into $s_2$ is not far from $\delta(|s_1|+|s_2|)$. This enables us to use the previously mentioned algorithms to find an approximately optimal window-compatible transformation, and show this is in fact a constant approximation away from $\delta(|s_1|+|s_2|)$.

\vspace{0.2cm}
{\noindent \textbf{Lemma} \ref{compatibility} [restated]. Given that $\edit(s_1,s_2) \leq \delta n$, there exists a window-compatible transformation of $s_1$ into $s_2$ with at most $(3\delta +  1/\gamma)n + 2l$ operations. \\}

Now we can put things in perspective. Lemma \ref{dp}, in light of the results of metric estimation, provides us a nice tool for finding an approximately optimal window-compatible transformation, and Lemma \ref{compatibility} argues that such a transformation is to some extent optimal. Based on this, we outline our algorithm for $\delta$-bounded edit distance as follows:
\begin{enumerate}
	\item Construct the windows of $W_1$ and $W_2$ for both $s_1$ and $s_2$.
	\item Construct a metric $\langle \metric, \edit\rangle$, where $\metric = W_1 \cup W_2$ and the distance of two points in $\metric$ is equal to the edit distance between their corresponding windows. We use the classic algorithm of edit distance to answer every oracle invocation for reporting the edit distance between two windows. Using the quantum approximation algorithm of metric estimation, find a $3+\epsilon$ approximation solution to the edit distances for every pair of windows (Theorem \ref{thm:metric0}). 
	\item Based on the estimated distances, find a $3+\epsilon$ approximately optimal window-compatible transformation (Lemma \ref{dp}).
	\item Report the transformation as an approximation proof for the $\delta$-bounded edit distance problem.
\end{enumerate}

We show in Section \ref{editdistance}, that by setting $\beta = 6/7$ and $\gamma = 1/\epsilon\delta$, the above algorithm runs in time $\tildorder(n^{2-1/7}\poly(1/\epsilon))$ and has an approximation factor of $7+\epsilon$. 

\vspace{0.2cm}
{\noindent \textbf{Lemma} \ref{mainbutnotmain} [restated] There exists a quantum algorithm that solves the $\delta$-bounded edit distance problem within an approximation factor of $7+\epsilon$ in time $\tildorder(n^{2-1/7}\poly(1/\epsilon))$. \\}

By Lemma \ref{mainbutnotmain}, we can approximate the $\delta$-bounded edit distance problem in truly subquadratic time in case the guarantee holds. Of course, if this algorithm provides a larger or invalid transformation, one can immediately imply that the guarantee $\edit(s_1,s_2) \leq \delta(|s_1|+|s_2|)$ is violated. The rest of the solution for edit distance follows from a simple multiplicative method. In order to solve edit distance, we first check whether the two strings are equal and in that case, we report that their distance is equal to $0$. Otherwise $\edit(s_1,s_2) \geq 1$. Now, we start with $\rho = 1/n$ and every time run our solution for $\delta$-bounded edit distance with parameter $\delta=\rho$, to find an approximation proof for $\edit(s_1,s_2) = \rho n$. If our algorithm finds a proper transformation with at most $(7\rho + \epsilon) n$ operations, then we report that solution. Otherwise, we know that $\edit(s_1,s_2) > \rho n$, and thus multiply $\rho$ by a factor $1+\epsilon$. Of course, this comes at the expense of an additional multiplicative factor of $1+\epsilon$ to the approximation factor; however, the running time remains $\tildorder(n^{2-1/7}\poly(1/\epsilon))$. We later refer to this technique as \textit{guess and multiply}.

\vspace{0.2cm}
{\noindent \textbf{Theorem} \ref{main} [restated] There exists a quantum algorithm that solves edit distance within an approximation factor of $7+\epsilon$ in time $\tildorder(n^{2-1/7}\poly(1/\epsilon))$. \\}

\subsection{Improving the Running Time via Bootstrapping}\label{bbc}
So far, we discussed how to use divide and conquer and metric estimation to approximate edit distance in subquadratic time. In this section, we explain the ideas to improve the running time of the algorithm by taking a hit on its approximation factor.

Recall that, in order to approximate the edit distance, we first construct a set of windows. Next, we use metric estimation to estimate the edit distances of the windows, and finally, we use a dynamic programming algorithm to find an almost optimal window-compatible transformation. As discussed before, such a solution approximates the edit distance within a constant factor. The components of this algorithm are illustrated in Figure \ref{7+}.

\begin{figure*}[h!]
\begin{center}
\includegraphics[scale=0.7]{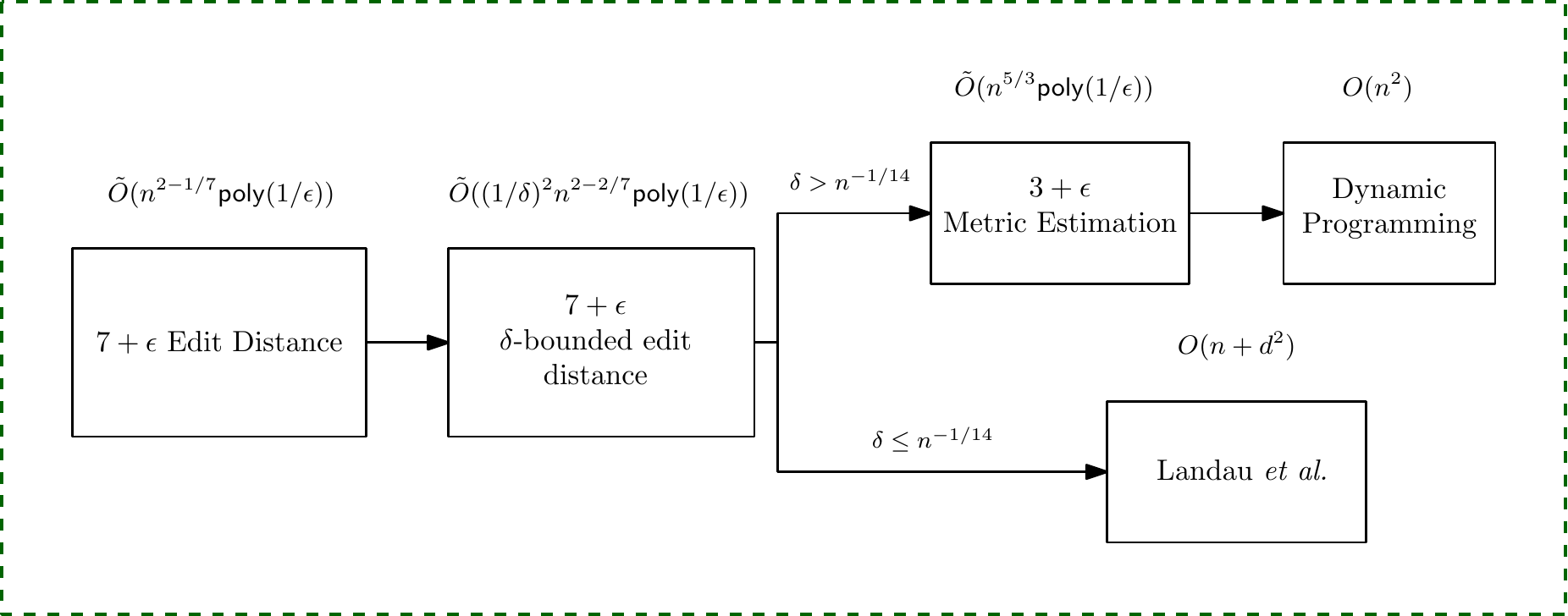}
\end{center}
\caption{The diagram depicts the components of the $7+\epsilon$ algorithm for edit distance. $x \rightarrow y$ shows that component $x$ uses component $y$ as a black box.}
\label{7+}
\end{figure*}

Now, we show that we can improve the algorithm at two points. Firstly, instead of using the $3+\epsilon$ approximation algorithm for metric estimation, we can lose a factor of $\errormetric(\epsilon)$ in the approximation and estimate the distances in time $\tildorder(n^{3/2+\epsilon}\poly(1/\epsilon))$ (Theorem \ref{thm:metric151}). In addition to this, as an oracle function for metric estimation, we do not really need to compute the exact edit distances of the windows; a constant estimation to the distances suffices. Therefore, one can use our algorithm for approximating edit distance to implement the oracle in subquadratic time. Of course, this again comes at the expense of deteriorating the approximation guarantee but the running time improves. In this section, we show how we combine these ideas to achieve an $\tildorder(n^{2-(5-\sqrt{17})/4+\epsilon}\poly(1/\epsilon))\simeq\tildorder(n^{1.781})$ time algorithm. As to why the exponent converges to $2-(5-\sqrt{17})/4$, we refer the reader to a discussion in Section \ref{bootstraping}.

To formalize the above ideas, suppose we are given two strings $s_1$ and $s_2$, and would like to approximate the edit distance between the strings in time $\tildorder(n^{2-(5-\sqrt{17})/4+\epsilon}\poly(1/\epsilon))$. We call our algorithm for this problem $\alg(\epsilon)$, and refer to its time complexity and approximation factor with $\timeedit(\epsilon) $ and $\erroredit(\epsilon)$, respectively. We inductively show that $$\timeedit(\epsilon) = \tildorder(n^{2-(5-\sqrt{17})/4+\epsilon}\poly(1/\epsilon))$$ and $\erroredit(\epsilon) = \errorvalue$. Notice that if $2-(5-\sqrt{17})/4+\epsilon \geq 2$, $\alg(\epsilon)$ can be trivially implemented with the classic $O(n^2)$ algorithm and the approximation factor $\erroredit(\epsilon) = 1$. Now, assume that $2-(5-\sqrt{17})/4+\epsilon < 2$.

An $\tildorder((1/\delta)^2n^{2-(5-\sqrt{17})/2+2\epsilon}\poly(1/\epsilon))$ time algorithm for $\delta$-bounded edit distance suffices to design $\alg(\epsilon)$. If $\delta \leq n^{-(5-\sqrt{17})/8+\epsilon/2}$ we run the $O(n+\delta^2n^2)$ of \landauincremental~\cite{landau1998incremental}, otherwise the running time of our algorithm is $\tildorder(n^{2-(5-\sqrt{17})/4+\epsilon}\poly(1/\epsilon))$. Moreover, a similar guess and multiply method explained in Section \ref{ghabli} extends this solution to edit distance. Therefore, all we need is to approximate the $\delta$-bounded edit distance problem in time $\tildorder((1/\delta)^2n^{2-(5-\sqrt{17})/2+2\epsilon}\poly(1/\epsilon))$. 
 To this end, we again define two parameters $\beta$ and $\gamma$ and set the window size equal to $\lfloor n^{1-\beta}\rfloor$ and the gap size equal to $g = \lfloor l/\gamma\rfloor$. Similar to what explained before, we construct two sets of windows $W_1$ and $W_2$ for $s_1$ and $s_2$ based on the windows size and gap size. Now, we use the same algorithm for finding the edit distance between $s_1$ $s_2$, with two modifications.

\begin{enumerate}
	\item Construct the windows of $W_1$ and $W_2$ for both $s_1$ and $s_2$.
	\item Construct a metric $\langle \metric, \edit\rangle$, where $\metric = W_1 \cup W_2$ and the distance of two points in $\metric$ is equal to the edit distance between their corresponding windows. We use $\alg(2\epsilon)$ (a slightly slower version of our algorithm) for estimating the edit distances of the windows in time $\timeedit(2\epsilon)=\tildorder(n^{2-(5-\sqrt{17})/4+2\epsilon}\poly(1/\epsilon))$ as on oracle function. Using the approximation algorithm of metric estimation, find an $\errormetric(\epsilon)\erroredit(2\epsilon)$ approximation solution to the edit distances for every pair of windows (Theorem \ref{thm:metric151}). 
	\item Based on the estimated distances, find an $\errormetric(\epsilon)\erroredit(2\epsilon)$ approximately optimal window-compatible transformation (Lemma \ref{dp}).
	\item Report the transformation as an approximation proof for the $\delta$-bounded edit distance problem.
\end{enumerate}

Notice that there are two modifications to the previous algorithm. First, instead of using the $3+\epsilon$ factor algorithm for metric estimation, here, we use an $\errormetric(\epsilon)$ approximation factor algorithm that runs in time $\tildorder(n^{3/2+\epsilon}\poly(1/\epsilon))$. Moreover, instead of implementing the oracle function via the classic $O(n^2)$ algorithm, we use $\alg(2\epsilon)$ for approximating the edit distances.
In Section \ref{bootstraping}, we show that by setting the right values for parameters $\beta$ and $\gamma$, the running time and approximation factor of algorithm $\alg(\epsilon)$ would be $\tildorder(n^{2-(5-\sqrt{17})/4+\epsilon}\poly(1/\epsilon))$ and $\erroredit(\epsilon) = \errorvalue$, respectively.

\vspace{0.2cm}
{\noindent \textbf{Theorem} \ref{mainbutnotmain} [restated] There exists an $\tildorder(n^{2-(5-\sqrt{17})/4+\epsilon})$ time quantum algorithm that approximates edit distance within a factor $\erroredit(\epsilon) = \errorvalue$. \\}

Figure \ref{figs:bootstrapping} shows the components of $\alg(\epsilon)$.
\begin{figure*}[h!]

\begin{center}
\includegraphics[scale=0.8]{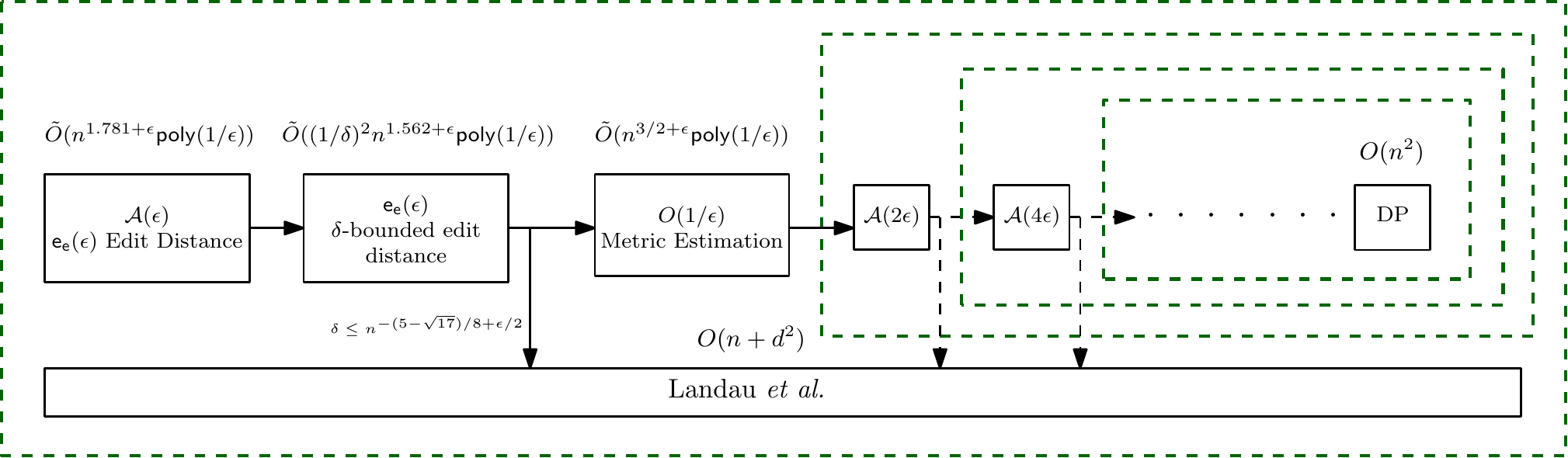}
\end{center}
\caption{The diagram illustrates the bootstrapping technique to achieve an $\tildorder(n^{1.781})$ time quantum algorithm for approximating edit distance. $x \rightarrow y$ shows that component $x$ uses component $y$ as a black box.}\label{figs:bootstrapping}
\end{figure*}

\section{Metric Estimation}\label{metric}
In this section, we discuss the metric estimation problem. Although the results of this section are only auxilary observations to be later used for edit distance, these results are of independent interest and may apply to future work. As defined previously, in this problem, we wish to estimate the distance matrix of a metric space $\metricspace{}$ with $n$ points. Notice that, an estimation of a distance $d(p_i, p_j)$ with approximation factor $\alpha$ lies in the range $[d(p_i, p_j), \alpha d(p_i, p_j)]$, therefore, the estimated value cannot be less than the actual distance. However, it can be more than the actual distance by a multiplicative factor of $\alpha$. We tend to minimize the query complexity and the approximation factor, however, our algorithm is allowed to run in time $\tildorder(n^2)$. Throughout this section, we show a tradeoff between the approximation factor and the quantum query complexity of metric estimation. First, we present an impossibility result that shows the approximation factor cannot be less than $3$ unless we make a quadratic number of queries. Next, in Section \ref{sub:alg1}, we present our desired $3+\epsilon$ approximation algorithm for metric estimation with a subquadratic query complexity. Afterward, we adjust our algorithm to make as few as $\tildorder(n^{3/2+\epsilon}\poly(1/\epsilon))$ oracle call for a larger constant approximation $\errormetric(\epsilon) = O(1/\epsilon)$.

\subsection{Hardness of Approximation for $\alpha < 3$}\label{sub:hardness}
As aforementioned, the purpose of this section is to show an impossibility result for approximating metric estimation within a factor smaller than $3$ with subquadratic query complexity. To this end, we give a reduction from the well-known parity problem to the metric estimation problem. Parity is one of the problems for which quantum computers cannot perform better than classical computers. Recall the definition of the parity problem from Section \ref{contribution:metric}.

\begin{center}
	\noindent\framebox{\begin{minipage}{6.3in}
			\textsf{Parity}\\[0.25cm]
			\textsf{Input}: an integer $n$, and access to an oracle $\oracle$ that upon receiving an integer $i$ reports the value of $f(i)$. $f$ is defined over $[n]$ and maps each index to either $-1$ or $1$.  \\[0.25cm]
			\textsf{Output}: $\parity(f) = \prod_{i \in [n]}f(i)$.
		\end{minipage}}
\end{center}

	Note that, $\mathsf{par}(f)$ is either $+1$ or $-1$ for every function $f$. \Farhi{}~\cite{Farhi:1998bz} proved that at least $\Omega(n)$ oracle queries are necessary to find $\mathsf{par}(f)$. A classic method to show lower bounds on the time/query complexity of problems is via a reduction from parity. This method has been used to show lower bounds on the quantum query complexity of many problems~\cite{doi:10.1137/050644719, Montanaro2015}.
We are now ready to present our reduction.

The idea is to construct a metric space from a given function $f$, and show that any estimation of the metric with an approximation factor smaller than $3$ can be used to compute the parity of $f$. A metric space should satisfy three properties: identity, symmetry and triangle inequality. Keep in mind that our construction should be in such a way that the metric meets all of the mentioned properties. For a function $f:[n^2]\rightarrow\{-1,1\}$, we construct a metric $\metric = \{a_1,a_2,\ldots,a_n,b_1,b_2,\ldots,b_n\}$ with $2n$ points. We divide the points into two groups, namely $a_i$'s and $b_i$'s, where the distances of the points within each group are all equal to $1$. Moreover, for every pair of points $(a_i,b_i)$, the distance of $a_i$ from $b_i$ is either $1/2$ or $3/2$, depending on function $f$. We show that, given an $\alpha < 3$ approximation estimation for the distances of $\metric$, one can determine $\parity(f)$ uniquely.

\begin{theorem}\label{hardness}
Any quantum algorithm that approximates the metric estimation problem with an approximation factor smaller than $3$ needs to make at least $\Omega(n^2)$ oracle calls.
\end{theorem}
\begin{proof}
As promised, we prove this theorem by reducing the parity problem  to the metric estimation problem. Suppose we are given an instance $\ii$ of the parity problem consisting of $f:[m] \rightarrow \{0,1\}$ and an oracle $\oracle$ to access $f$. We assume w.l.o.g that $m = n^2$ and construct an instance $\cor(\ii)$ of metric estimation as follows:
let $\langle \metric, d\rangle$ be a set of $2n$ points where the distance of the points $p_i$ and $p_j$ is denoted by $d(p_i,p_j)$. We divide the points of the metric into two groups $\{a_1,a_2,\ldots,a_n\}$ and $\{b_1,b_2,\ldots,b_n\}$. As mentioned before, the distances within the points of each group are equal to $1$. Moreover, for every pair of points $a_i$ and $b_j$, we set $d(a_i,b_j)$ as follows:
\begin{equation*}
d(a_i,b_j) =
\begin{cases}
3/2 & f((i-1)n+j) = 1, \\
1/2 & \text{otherwise.}
\end{cases}
\end{equation*}
The identity and symmetry conditions are met by definition. We show that the triangle inequality also holds. If all three points of a triangle are in the same group (either $a_i$'s or $b_i$s), then their distances are all $1$. If they are in different groups, the distances are one of these cases, $\langle1, 1/2, 1/2\rangle$, $\langle1, 1/2, 3/2\rangle$ or $\langle1, 3/2, 3/2\rangle$, all of which meet the triangle inequality. Thus, $\metricspace{}$ is a valid metric space. One can trivially construct an oracle $\oraclep$ for $\cor(\ii)$, that reports the distance of a pair of points with a single query to $\oracle$.

Now, suppose for the sake of contradiction that there exists a quantum algorithm that estimates the distances within a factor smaller than 3 with $o(n^2)$ query calls of $\oraclep$.
We show we can use this algorithm to find $\mathsf{par}(f)$ as follows. We first run the algorithm to approximate all of the distances via $o(n^2)$ query calls to $\oraclep$. This costs us a total of $o(n^2)$ queries to $\oracle$ since every query of $\oraclep$ makes a call to $\oracle$. Next, for every pair of points $(a_i,b_i)$ we determine $f((i-1)n+j)$ as follows:

\begin{equation*}
f((i-1)n+j) =
\begin{cases}
1 & d^*(a_i,b_j) \geq 3/2. \\
-1 & \text{otherwise}
\end{cases}
\end{equation*}
where $d^*(a_i,b_j)$ is the estimated distance of point $a_i$ from point $b_j$. The correctness of our reduction follows from the fact that the approximation factor of the algorithm for metric estimation is smaller than $3$ and thus if $d^*(a_i,b_j) \geq 3/2$ the actual distance $d(a_i,b_j)$ is more than $1/2$. Finally, we take the multiplication of all determined values for $f$ and compute $\parity(f)$ with $o(n^2) = o(m)$ queries. This contradicts the observation of \Farhi{}~\cite{Farhi:1998bz}.
\end{proof}

\subsection{A $3+\epsilon$ Approximation Algorithm with $\algorder$ Queries}\label{sub:alg1}
In this section, we present a quantum algorithm to estimate the distances of a metric space within an approximation factor of $3+\epsilon$. Our algorithm makes $\algorder$ oracle calls.

The first idea of our algorithm is to discretize the distances. Recall that, the distances of the metric are non-negative integers in the interval $[l,u]$. We separate the numbers into disjoint intervals. If $l = 0$, we put a separate interval $[0,0]$ for $0$ and continue on with the numbers in $[1,u]$. Every time, we find the smallest number $l \leq x \leq u$ which is not covered in the previous intervals and add a new interval $[x, (1+\epsilon)x]$ to the list. Since $u = \poly(n)$, the number of intervals is $\poly(\log n)\poly(1/\epsilon) = \tildorder(\poly(1/\epsilon))$. Now, by losing a factor $1+\epsilon$ in the approximation, we can round up all of the numbers within an interval to its highest value and solve the problem for each interval separately. Therefore, the problem boils down to the following: given a threshold $t$, find all pairs of the points with a distance of at most $t$. We call this problem \thresholdestimation{}. Note that, since we wish to find a $3$ approximation solution for \thresholdestimation{}, a false positive is also allowed in the solution. More precisely, the solution should contain all pairs of points within a distance of at most $t$, but pairs within distances up to $3t$ are also allowed to be included.

In order to approximate \thresholdestimation{}, we subsequently make use of Grover's search algorithm~\cite{boyer1996tight}. Think of the metric as a graph $G$ where every point corresponds to a vertex of the graph and two vertices are adjacent if the distance of their corresponding points is at most $t$.
Let $0 < \degreet < 1$ be a fixed parameter. We call a vertex $v$ of the graph \textit{low degree} if the number of edges incident to $v$ are bounded by $n^\degreet$ and \textit{high degree} otherwise. Our algorithm deals with low degree vertices and high degree vertices differently. We set the value of $\degreet$ after the analysis and show it gives us the best bound.

In our algorithm, we iterate over the vertices of the graph and find their neighbors one by one. To this end, fix a vertex $v_i$ and suppose we wish to find all of its neighbors. Due to Grover's search (Theorem \ref{grover}), we can list up to $n^\degreet$ neighbors of $v_i$ with $\sqrt{n^\degreet n} = n^{(1+\degreet)/2}$ queries. Moreover, with an additional Grover's search, we can determine whether the degree of $v_i$ is more $n^\degreet$ with $O(\sqrt{n})$ queries. 
If $v_i$ is low degree, we already have all its neighbors, and thus we can report those edges and remove $v_i$ from the graph. Otherwise, the degree of $v_i$ is more than $n^\degreet$. In this case, we make $O(n)$ oracle calls and find the distances of all other points from the corresponding point of $v_i$, namely $p_i$. Based on these distances, we construct two sets of vertices $N(v_i, t)$ and $N(v_i, 2t)$ where the former contains all vertices corresponding to points within a distance of at most $t$ of $p_i$ and the latter contains all of the vertices corresponding to points within a distance of at most $3t$ from $p_i$. We then proceed by reporting all the edges between $N(v_i, t)$ and $N(v_i, 2t)$ and removing $N(v_i, t)$ from the graph.
A pseudocode for this algorithm is shown in Algorithm \ref{alg:metric0}.

\begin{algorithm2e}
 \KwData{
 The number of points in the metric space $\metric = \{p_1,p_2,\ldots,p_n\}$, oracle access to the distances between points, and a threshold $t$.}
 \KwResult{A 0-1 matrix $A$ of size $n\times n$, where for each $d(p_i,p_j)\leq t$ we have $A_{i,j}=1$, and for each $A_{i,j}=1$ we have $d(p_i,p_j)\leq 3t$.}
Initialize a graph $G$ with $n$ vertices\;
 \While{$V(G)$ is not empty}{
Select a vertex $v_i$ from $V(G)$\;
List up to $n^\degreet$ neighbors of $v_i$ and find out whether $v_i$ is high degree or low degree\;
\If{$v_i$ is low degree}{
Update the matrix $A$ according to the edges of $v_i$\; 
Remove $v_i$ from $V(G)$\;
}
\Else{
	Find the distances of $p_i$ from all other points\;
	Construct $N(v_i, t)$ and $N(v_i, 2t)$ based on the distances\;
	For every $x \in N(v_i, t)$ and $y \in N(v_i, 2t)$, set $A_{x,y} = 1$\;
	$V(G) \leftarrow V(G) \setminus N(v_i,t)$\; 
}
} 
  Output $A$\;
 \caption{\textsf{EstimateWithThreshold}($n, \oracle, t$)}
 \label{alg:metric0}
\end{algorithm2e}

\begin{theorem}\label{thm:metric0}
For $\degreet{} = 1/3$, Algorithm \ref{alg:metric0} approximates \thresholdestimation{} within a factor of $3$ with $O(n^{5/3})$ oracle calls. Moreover, the running time of Algorithm \ref{alg:metric0} is $O(n^2)$.
\end{theorem}
\begin{proof}
The correctness of our algorithm follows from the triangle inequality. We first show that for every pair of points $p_i$ and $p_j$ such that $d(p_i,p_j) \leq t$, $A_{i,j} = 1$ at the end of the algorithm. To this end, consider the first time that we remove either $v_i$ or $v_j$ from the vertices. This could happen in two ways: either one of $v_i$ or $v_j$ is removed from the graph as a low degree vertex or any of them is removed in an iteration of the algorithm for some high degree vertex. In the former case, since we find all neighbors of the low degree vertices, we detect the edge between them thus $A_{i,j} = 1$. Now, suppose that one of these vertices say $v_i$ is removed from the graph in an iteration for a vertex $v_x$ of the graph. Therefore, $d(v_i, v_x) \leq t$. Moreover, due to the triangle inequality, $d(v_j,v_x) \leq d(v_j,v_i) + d(v_i,v_x) \leq 2t$ and thus $v_j \in N(v_x,2t)$. Thus we set $A_{i,j} = 1$. Moreover, it follows from the triangle inequality that if we set $A_{i,j} = 1$ for some $i$ and $j$, then the distance of the points $p_i$ and $p_j$ is bounded by $3t$.

Trivially, the running time of the algorithm is $O(n^2)$. In what follows we show the query complexity of the algorithm is bounded by $O(n^{5/3})$. Let $Q(n)$ denote the query complexity of the algorithm for the case where $|V(G)| = n$. To compute $Q(n)$, we consider two cases separately: (i) when we select a vertex $v_i$ which is low-degree and (ii) when we select a vertex $v_i$ which is high degree. In any case, we make a search to list up to $n^\degreet$ neighbors of $v_i$ and we make at least $O(n^{(1+\degreet)/2})$ oracle calls. In addition to this, we make $O(\sqrt{n})$ more oracle calls to find out whether $v_i$ is low degree. In case $v_i$ is low degree, we remove $v_i$ from the graph and continue on with an instance with $n-1$ vertices. Otherwise, we make $O(n)$ more oracle calls and then remove $N(v_i,t)$ from the graph which leaves us an instance with at most $n-n^\degreet$ vertices. Therefore, we formulate $Q(n)$ as follows:

\begin{equation*}
Q(n) = 
\begin{cases}
O(n^{(1+\degreet)/2}) + O(\sqrt{n})+ Q(n-1)  
&\text{if }v_i\text{ is low degree,} \\
O(n^{(1+\degreet)/2}) + O(\sqrt{n})+ O(n) + Q(n-n^\degreet)
& \text{otherwise.}
\end{cases}
\end{equation*}
Now we set $\degreet = 1/3$ and thus we obtain 
\begin{equation*}
Q(n) = 
\begin{cases}
O(n^{2/3}) + O(\sqrt{n}) + Q(n-1) = O(n^{2/3}) + Q(n-1)   & \text{if }v_i\text{ is low degree,} \\
O(n^{2/3}) + O(\sqrt{n}) + O(n) + Q(n-n^{1/3}) =  O(n) + Q(n-n^{1/3}) & \text{otherwise.}
\end{cases}
\end{equation*}
A trivial analysis shows that for every vertex that we remove from $V(G)$, we make $O(n^{2/3})$ amortized query calls and thus the total number of queries is bounded by $n\cdot O(n^{2/3}) = O(n^{5/3})$.
\end{proof}

Now, we are ready to present our $3+\epsilon$ approximation algorithm with query complexity $\algorder$.
For each $i$, using Algorithm \ref{alg:metric0}, we can find all distances in range $[0, l(1+\epsilon/3)^{i+1}]$ with some false positive distances in range $[l(1+\epsilon/3)^{i+1}, 3l(1+\epsilon/3)^{i+1}]$. By knowing the same information for $i-1$, we have all points in range $[0, l(1+\epsilon/3)^{i}]$ with some false positive distances in range $[l(1+\epsilon/3)^{i}, 3l(1+\epsilon/3)^{i}]$. Thus we can find all points in range $[l(1+\epsilon/3)^{i}], l(1+\epsilon/3)^{i+1}]$, some false positives in range $[l(1+\epsilon/3)^{i+1}, 3l(1+\epsilon/3)^{i+1}]$, and some false negatives that estimated correctly before. All of these distances are in range $[l(1+\epsilon/3)^{i}, 3l(1+\epsilon/3)^{i+1}]$. Therefore we can estimate these distances as $3l(1+\epsilon/3)^{i+1}$ and the approximation factor is $\frac{3l(1+\epsilon/3)^{i+1}}{l(1+\epsilon/3)^{i}} = 3(1+\epsilon/3) = 3+\epsilon$. The time and query complexity of this algorithm is the time and query complexity of Algorithm \ref{alg:metric0} times $\log_{1+\epsilon/3}(\upperdistance/\lowerdistance) = \tildorder(1/\epsilon)$. We handle zero distances separately. You can find the pseudocode of this algorithm in the following.

\begin{algorithm2e}
 \KwData{The number of points in the metric space $\metric=\{p_1,p_2,\ldots,p_n\}$, oracle access to the distances between points, a small number $\epsilon > 0$, a lower bound, and an upper bound for the distances.}
 \KwResult{An $n\times n$ matrix $A$, where $A_{i,j}$ is a $3+\epsilon$ approximation of $d(p_i, p_j)$}
Initialize three matrices $A$, $A^\circ$ and $A^\bullet$\;
$A^\circ \gets \textsf{EstimateWithThreshold}(n, \oracle, 0)$\;
Initialize the threshold: $t \leftarrow \max(1,\lowerdistance)$\;
\While{$t \leq \upperdistance$}{
$t\leftarrow t\cdot (1+\epsilon/3)$\;
$A^\bullet \gets $\textsf{EstimateWithThreshold}($n, \oracle, t$)\;
$A\gets A + (A^\bullet - A^\circ)\cdot 3t$\;
$A^\circ\gets A^\circ\vee A^\bullet$
  }
  output A
 \caption{\textsf{EstimateMetric}($n, \oracle, \epsilon, \lowerdistance, \upperdistance$)}
 \label{alg:metric1}
\end{algorithm2e}

\begin{theorem}\label{thm:metric1}
Algorithm \ref{alg:metric1} solves metric estimation problem with approximation factor $3+\epsilon$, quantum query complexity $\tildorder(n^{5/3})$ and time complexity of $\tildorder(n^2)$ for an arbitrary small constant $\epsilon>0$.
\end{theorem}
\begin{proof}
The correctness of Algorithm \ref{alg:metric1} follows from that of Algorithm \ref{alg:metric0}. Moreover, Algorithm \ref{alg:metric1} runs Algorithm \ref{alg:metric0}, $\tildorder(\poly(1/\epsilon))$ times; therefore, the query complexity of Algorithm \ref{alg:metric1} is $\tildorder(n^{5/3+\epsilon}\poly(1/\epsilon))$. Furthermore, the running time of Algorithm \ref{alg:metric1} is $\tildorder(n^2\poly(1/\epsilon))$.
\end{proof}

In this section, we achieved an algorithm with subquadratic query complexity and approximation factor $3+\epsilon$ for any $\epsilon>0$ which is nearly optimal due to Theorem \ref{hardness}. In Section \ref{metric15}, we reduce the quantum query complexity to $O(n^{3/2+\epsilon})$, but the approximation factor grows to larger constants.

\subsection{A Constant Approximation Algorithm with $\tildorder(n^{3/2+\epsilon}\poly(1/\epsilon))$ Queries}\label{sub:alg2}
\label{metric15}
In Sections \ref{sub:hardness} and \ref{sub:alg1}, we showed that the best approximation factor that we can get with subquadratic oracle calls are bounded from below by $3$ and that a $3+\epsilon$ approximation is possible. In this section, we complement this result by showing that the query complexity can be further reduced to $\tildorder(n^{3/2+\epsilon}\poly(1/\epsilon))$, and moreover, we show that the required query complexity is at least $\Omega(n^{3/2})$ for any constant approximation factor.
To this end, we present a quantum algorithm with expected query complexity $\tildorder(n^{3/2+\epsilon}\poly(1/\epsilon))$ where the approximation factor and the expected running time are $\errormetric(\epsilon) = O(1/\epsilon)$ and $\tildorder(n^2\poly(1/\epsilon))$, respectively.

As stated before, the problem reduces to \thresholdestimation. Similar to what we did for Theorem \ref{thm:metric1}, we divide the vertices into two categories low degree and high degree.
Low degree vertices are easy to deal with; we simply list all of their neighbors using Grover's search and report all of them. If a vertex is high degree though, the algorithm needs to be more intelligent.

The overall idea is summarized in the following: we find a small group of vertices, namely \textit{representatives}, that hits at least one vertex from the neighborhood of any large degree vertex.
Using a standard argument of hitting sets, we can show that a subset of $\tildorder(n/\eta)$ vertices chosen uniformly at random, as \textit{representatives}, hits every neighborhood of size at least $\eta$ with high probability. Notice that these neighborhoods are at most $n$ fixed but unknown subsets. Other vertices outside \textit{representatives} are either low degree vertices, or \textit{followers} which have at least one neighbor in \textit{representatives}, or both.
Next, we run the following procedure: for every vertex $v_i$ which is not in \textit{representatives}, we first check if it is a follower. For a follower vertex which has at least one neighbor in \textit{representatives}, we select one such vertex and call that the leader of $v_i$. Otherwise, if there is no such neighbor, we conclude that $v_i$ is indeed low degree; thus we can find all its neighbors via Grover's search and update the solution. Next, we solve the problem recursively for all of the representatives.
For any $v_i$ and $v_j$ which are connected, we want the leader of $v_i$ and the leader of $v_j$ to become connected in the recursive result. As a consequence of the triangle inequality, we can achieve this by tripling the threshold.
Finally, we construct our solution based on the approximated solution of the representatives and the leader-follower relations, simply by connecting any two vertices, where their leaders are connected. The approximation factor increases with each recursion, but since the number of recursions is a constant, we achieve a constant approximation factor. Furthermore, in each recursion call, we can increase the degree threshold as far as it doesn't increase the query complexity too much. By increasing the degree threshold to its 3rd power, we have this property.
The number of vertices in nested recursions depleted, as soon as the degree threshold become larger that the number of vertices, in which case we treat all vertices as low degree, thus the next time we have zero vertices and the process finishes.

The pseudocode of the algorithm is shown below.

\begin{algorithm2e}
\KwData{The number of points in the metric space $\metric=\{p_1,p_2,\ldots,p_n\}$, oracle access to the distances between points, a threshold $t$, a small number $\epsilon$, and a degree threshold $\degreethreshold{}$}
\KwResult{An $n\times n$ matrix $A$, where $A_{i,j}$ is an $\errormetric(\epsilon)$ approximation of $d(p_i, p_j)$.}
 \If {$n=0$} {
  Output an empty matrix\;
 } \Else {
   Sample a hitting set $\mathcal{R}$ with $O((n/\degreethreshold{})\log n)$ points\;
   Initialize an $n\times n$ matrix $A$\;
    \For {all points in $\metric$ as $v_i$} {
	Find a neighbor of $v_i$ or $v_i$ itself in $\mathcal{R}$ and save it as $l(v_i)$ (the leader of $v_i$)\;
	\If {no such neighbor of $v_i$ exists and $v_i$ is not in $\mathcal{R}$} {
		List all neighbors of $v_i$\;
	}
}
$A' \gets \mathsf{FastEstimateWithThreshold}(\mathcal{R}, \oracle, 3t, 3\epsilon, (\degreethreshold{})^3$)\;
\For {all pairs of points in $\metric$ as $(v_i, v_j)$ where $l(v_i)\neq\varnothing$ and $l(v_j)\neq\varnothing$} {
	\If {$A'(l(v_i), l(v_j)) = 1$} {
		$A(v_i, v_j)\gets 1$\;
	}
}
$A\gets A \vee A'$\;
Output $A$\;
	}
 \caption{\textsf{FastEstimateWithThreshold}($\metric, \oracle, t, \epsilon, \degreethreshold{}$)}
 \label{alg:metric150}
\end{algorithm2e}

\begin{theorem}
\label{thm:metric150}
Algorithm \ref{alg:metric150} called with the threshold $t$, the parameter $\epsilon$ and the degree threshold $n^{2\epsilon}$ finds all distances less than $t$ with some false positive distances in range $[t, \errormetric(\epsilon)\cdot t]$ where $\errormetric(\epsilon) = O(1/\epsilon)$, in expected query complexity $\tildorder(n^{3/2+\epsilon})$ and expected time complexity $\tildorder(n^2)$.
\end{theorem}

\begin{proof}
As aforementioned, we deal with three groups of vertices: \textit{representatives}, \textit{followers} and low degree vertices. Low degree vertices may intersect with the other two, but each vertex is at least in one group. Here, for any low degree vertex outside other two groups, we find its neighborhood explicitly. Therefore, to show the correctness of the algorithm, we focus on two groups of followers and representatives. First, we show that we correctly find the group of representatives. A subset $\mathcal{R}$ of size $2(n/\degreethreshold{})\ln n$ chosen uniformly at random, misses one fixed neighborhood of size at least $\degreethreshold{}$ with a probability of at most $(1-\frac{\degreethreshold{}}{n})^{2(n/\degreethreshold{})\ln n}\approx 1/\mathrm{e}^{2\ln n}=1/n^2$. For all neighborhoods, which are at most $n$ fixed subsets of size at least $\degreethreshold{}$, the probability of missing at least one neighboorhood is at most $n\cdot (1 /{n^2}) = 1/n$ by the union bound. If $\mathcal{R}$ misses at least one large neighborhood, we can reset the algorithm. A standard argument of Las Vegas algorithms ensures that the expected query complexity and expected running time is no more that $\frac{n}{n-1}$ times the query complexity and running time of one execution, respectively (Exercise 1.3 of \cite{motwani1995randomized}). Now we can continue assuming we have leader-follower relations.
Recall that for every follower $v_i$ we select one of its neighbors in $\mathcal{R}$ and call that vertex the leader of $v_i$. To simplify the last part of the algorithm, for any $v_i$ in $\mathcal{R}$, we call $v_i$ as the leader of itself. Thus, all followers and representatives have leaders.

Furthermore, we solve the problem for the group of representatives recursively, with different parameters. We triple the threshold in each recursion. Call the leader of two connected vertices $v_i$ and $v_j$ as $r_i$ and $r_j$, respectively. By the triangle inequality we have $d(r_i,r_j) \leq d(r_i,v_i)+d(v_i,v_j)+d(v_j,r_j) \leq 3t$. Thus the leader of any two connected vertices is connected by the new threshold; hence we find all distances less than $t$, perhaps with some false positives.

Before we compute the approximation factor $\errormetric(\epsilon)$, we determine the number of nested recursion calls. We call the number of vertices in the $i$'th recursion call $n_i$. Note that $n_i$ is the size of representatives group of the $(i-1)$'th recursion. Thus, we have $n_i = O((n_{i-1}/\degreethreshold{i})\log n_{i-1})$. Using induction, we can show that the degree threshold in $i$'th recursion call is $\degreethreshold{i} = n^{(2\cdot 3^{i})\epsilon}$ and therefore, $n_i = O(n^{1-(3^{i}-1)\epsilon}\cdot O(\log^i(n)))$. The number of vertices becomes zero in $i$'th recursion, where $1-(3^{i}-1)\epsilon<0$ or $i > \log_3(1+1/\epsilon)$. Hence, we have at most $k(\epsilon) = \log_3(1/\epsilon)+1$ nested recursion calls, which is independent of $n$. Notice that $k(3\epsilon) = k(\epsilon) - 1$ and $k(3^{k(\epsilon)}\epsilon) = 0$.

What is remained is to compute the approximation factor $\errormetric(\epsilon)$. The maximum distance of a pair of vertices that we report an edge between them is at most $2(1+3\errormetric(3\epsilon))$ times the threshold. We know that $\errormetric(3^{k(\epsilon)}\epsilon) = 1$, therefore $\errormetric(\epsilon)\leq 9/\epsilon = O(1/\epsilon)$.

The query complexity of Grover's search in the $i$'th recursion is at most $n_i O(\sqrt{|\mathcal{R}|})$ to find the leader of each point, plus $n_i O(\sqrt{n_i\cdot \degreethreshold{i}})$ to find all neighbors of some low degree points. This is equal to $O(n^{3/2-(\frac{5\cdot 3^{i}-3}{2})\epsilon}\polylog(n)) + O(n^{3/2-(\frac{3^{i}-3}{2})\epsilon}\polylog(n))$. Notice that the latter term dominates the former, and the query complexity for $i=0$ dominates all of the recursions; therefore the query complexity of Algorithm \ref{alg:metric150} is at most $O(n^{3/2+\epsilon})$.

The time complexity is at most $O(n^2)$ in each phase. Thus, the time complexity is $O(n^2\polylog(n))$.
\end{proof}

In what follows, we complete our algorithm using Algorithm \ref{alg:metric150} with several thresholds. This is the same as Algorithm \ref{alg:metric1} with minor differences such as line 8 where $3$ has been replaced with $\errormetric(\epsilon)$.

\begin{algorithm2e}
 \KwData{The number of points in the metric space $\metric$, oracle access to the distances between points, a small number $\epsilon>0$, a lower bound, and an upper bound for the distances}
 \KwResult{An $n\times n$ matrix $A$, where $A_{i,j}$ is a $\errormetric(\epsilon)$ approximation of $d(p_i, p_j)$ in $\metricspace$}
Initialize the distance estimation matrix $A$, $A^\circ$ and $A^\bullet$\;
$A^\circ \gets$\textsf{FastEstimateWithThreshold}($n, \oracle, 0, \epsilon, n^{2\epsilon}$)\;
Initialize the threshold: $t \leftarrow \max(1,\lowerdistance)$\;
\While{$t \leq \upperdistance$}{
$t\leftarrow t\cdot (1+\epsilon)$\;
$A^\bullet \gets $\textsf{FastEstimateWithThreshold}($n, \oracle, t, \epsilon, n^{2\epsilon}$)\;
	$A\gets A + (A^\bullet - A^\circ)\cdot \errormetric(\epsilon)$\;
$A^\circ\gets A^\circ\vee A^\bullet$\;
  }
  Output $A$\;
 \caption{\textsf{FastEstimateMetric}($\metric, \oracle, \epsilon, \lowerdistance, \upperdistance$)}
 \label{alg:metric151}
\end{algorithm2e}

\begin{theorem}
\label{thm:metric151}
Algorithm \ref{alg:metric151} solves the metric estimation with approximation factor $\errormetric(\epsilon) = O(1/\epsilon)$, with query complexity $\tildorder(n^{3/2+\epsilon}\poly(1/\epsilon))$ in time $\tildorder(n^2 \poly(1/\epsilon))$.
\end{theorem}
\begin{proof}
The correctness of Algorithm \ref{alg:metric151} follows from that of Algorithm \ref{alg:metric150}, the same as we did in Theorem \ref{thm:metric1}. Moreover, Algorithm \ref{alg:metric151} runs Algorithm \ref{alg:metric150}, $\tildorder(\poly(1/\epsilon))$ times; therefore, the query complexity of Algorithm \ref{alg:metric1} is $\tildorder(n^{3/2+\epsilon}\poly(1/\epsilon))$. Furthermore, the running time of Algorithm \ref{alg:metric1} is $\tildorder(n^2 \poly(1/\epsilon))$.
\end{proof}

\subsection{An $\Omega(n^{3/2})$ Time Lower Bound}
Last but not least, we show that the query complexity of metric estimation cannot be reduced any further, so long as the approximation factor is constant, i.e., we need at least $\Omega(n^{3/2})$ queries to approximate metric estimation within a constant factor. We use Ambainis's lower bound technique \cite{Ambainis:2000:QLB:335305.335394}. 

\begin{theorem}[proven in~\cite{Ambainis:2000:QLB:335305.335394}, Theorem 6]\label{amba}
Let $f(x_1, \ldots, x_n)$ be a function of $n$ variables 
with values from some finite set
and $X, Y$ be two sets of inputs such that $f(x)\neq f(y)$
if $x\in X$ and $y\in Y$.
Let $R\subset X \times Y$ be such that
\begin{enumerate}
\item
For every $x\in X$, there exist at least $m$ different $y\in Y$ such that
$(x, y)\in R$.
\item
For every $y\in Y$, there exist at least $m'$ different $x\in X$ such that
$(x, y)\in R$.
\end{enumerate}
Let $l_{x, i}$ be the number of $y\in Y$ such that $(x, y)\in R$ and $x_i\neq y_i$
and $l_{y, i}$ be the number of $x\in X$ such that $(x, y)\in R$ and $x_i\neq y_i$.
Let $l_{max}$ be the maximum of $l_{x, i}l_{y, i}$ over all $(x, y)\in R$
and $i\in\{1, \ldots, N\}$ such that $x_i\neq y_i$.
Then, any quantum algorithm computing $f$ uses  
$\Omega(\sqrt{\frac{m m'}{l_{max}}})$ queries.
\end{theorem}

Now we use an intermediate problem to prove the desired lower bound. A permutation matrix is a boolean $n\times n$ matrix, which has exactly one entry $1$ in each row and each column. It corresponds to a permutation $\pi$ where entries of 1 are in the form of $(i, \pi(i))$. The sign of a permutation matrix is defined as the sign of its corresponding permutation. The next lemma about the problem of determining the sign of a permutation matrix is the main part of out lower bound.

\begin{lemma}
\label{perm}
Any quantum algorithm which takes an $n\times n$ permutation matrix as the input and outputs the sign of the permutation matrix has a query complexity of at least $\Omega(n^{3/2})$.
\end{lemma}
\begin{proof}
To apply Theorem \ref{amba}, we use a single index to address an entity instead of two indices. Assume $f(x_1, x_2, . . ., x_{n^2})$ is a function which takes a permutation matrix as input and outputs a value in $\{-1, 1\}$ as the sign of the matrix. Define $X$ as the set of permutation matrices with sign $-1$, $Y$ as the set of permutation matrices with sign $1$ and $R\subset X\times Y$ such that $(x,y)\in R$ iff their corresponding matrices can be transformed to the other with a swap of just two rows. Therefore, we have $m=m'=\binom n 2$. For an $i$ we have $l_{x,i}=n-1$ and $l_{y,i}=1$ if $x_i=1$ and $l_{x,i}=1$ and $l_{y,i}=n-1$ if $x_i=0$, thus $l_{max} = n-1$. Therefore by Theorem \ref{amba}, every quantum algorithm to solve this problem has a query complexity of at least $\Omega\Big(\sqrt{\frac{{\binom n 2}^2}{n-1}}\Big) = \Omega(n^{3/2}).$
\end{proof}

The problem of determining the sign of an $n\times n$ permutation matrix can be easily reduced to our problem, by constructing a bipartite graph with parts $X$ and $Y$, $n$ vertices in each part and $n$ edges that form a complete matching between $X$ and $Y$. Every matching has a corresponding permutations and vice versa. Therefore, we have the following theorem.

\begin{theorem}
\label{thm:lowerquery}
Any quantum algorithm which estimates distances of a metric space of $n$ points with a constant approximation factor has a query complexity of at least $\Omega(n^{3/2})$.
\end{theorem}
\begin{proof}
We simply reduce the problem of determining the sign of a permutation matrix to this problem. Assume $n$ is an even number. For an instance of a $n/2\times n/2$ permutation matrix $\mathcal{A}$, we construct a metric space $\metric$ with $n$ points, $r_i$ for row $i$ and $c_j$ for column $j$ of the matrix. Make the distance between $r_i$ and $c_j$ equal to $1$ where $\mathcal{A}_{i,j}=1$ and a distance of $n^2$ otherwise. The distances meet the necessary conditions. Notice that we do not construct the distances, we construct an oracle which invokes the oracle of $\mathcal{A}$ at most one time. Using Lemma \ref{perm}, the query complexity is at least $\Omega((n/2)^{3/2})=\Omega(n^{3/2})$.
\end{proof}

\section{Edit Distance}\label{editdistance}
In this section, we use the results of Section \ref{metric} to design a quantum approximation algorithm for the edit distance problem. Our algorithm has an approximation factor of $7+\epsilon$ for an arbitrarily small number $\epsilon>0$ and time complexity $\tildorder(n^{2-1/7}\poly(1/\epsilon))$. The outline of the algorithm is presented in Section \ref{section:ourResults}. Here we provide detailed proofs of the lemmas and theorems that are previously used for edit distance.

\begin{lemma}\label{dp}
Given a matrix of edit distances between the substrings corresponding to every pair of windows of $W_1$ and $W_2$, one can compute the optimal window-compatible transformation of $s_1$ into $s_2$ in time $O(n+|W_1||W_2|)$.
\end{lemma}

\begin{proof}
We take a dynamic programming approach to find the optimal window-compatible transformation of the two strings. Suppose $W_1$ has $k$ and $W_2$ has $k'$ windows. We recall that $W_1 = \langle w_1, w_2, . . ., w_k\rangle$ and $W_2=\langle w'_1, w'_2, . . ., w'_{k'}\rangle$ are collections of windows for $s_1$ and $s_2$, respectively, with window size $l$ and gap size $g$. We note that every window $w$ corresponds to a subinterval of $[1,n]$, thus we can use a linear time sorting algorithm, such as bucket-sort, and sort all the windows in time $O(n)$. Therefore, we can assume that the windows in $W_1$ and $W_2$ are sorted according to their right side.

Now for every $1\leq i\leq k$ and $1\leq j\leq k'$ we define $c_{i,j}$ to denote the optimal window-compatible transformation of the first $l+(i-1)\cdot g$ characters of $s_1$ into the first $l+(j-1)\cdot g$ characters of $s_2$. These suffixes of $s_1$ and $s_2$ correspond to $\langle w_1, w_2, . . ., w_i\rangle$ and $\langle w'_1, w'_2, . . ., w'_j\rangle$, respectively. For the sake of simplicity, we define $c_{i,0}=l+(i-1)\cdot g$, which is the cost of only deleting, and $c_{0,j}=l+(j-1)\cdot g$, which is the cost of only inserting. For every $1\leq i\leq k$ and $1\leq j\leq k'$ the following recursive formula holds:

\begin{align}
	c_{i,j} = \min \Big \{ &c_{i-1,j}+g, \enspace c_{i,j-1}+g, 
	 c_{i-\lceil l/g \rceil, j-\lceil l/g \rceil}+d(w_i, w'_j) \Big \}\enspace.
\label{eq:windowdynamic}
\end{align}

To compute $c_{i,j}$ we have three possibilities in an optimal window-compatible transformation. In particular, either $w_i$ is matched with $w'_j$, or at least one of $w_i$ and $w'_j$ is unmatched. If we match $w_i$ to $w'_j$, then there is a cost of $d(w_i, w'_j)$ for transforming $w_i$ to $w'_j$. Also, the other windows that overlap with $w_i$ or $w'_j$ cannot be used. Consequently, the problem reduces to finding window-compatible transformations for the first $l+(i-1)\cdot g-l$ characters of $s_1$ and the first $l+(j-1)\cdot g-l$ characters of $s_2$ with respect to $\langle w_1, w_2, . . ., w_{i-\lceil l/g \rceil}\rangle$ and $\langle w'_1, w'_2, . . ., w'_{j-\lceil l/g \rceil}\rangle$. This subproblem is captured by $c_{i-\lceil l/g \rceil, j-\lceil l/g \rceil}$. For the case that $w_i$ is unmatched, we need $g$ operations to remove every character in range $l+(i-2)\cdot g+1,\ldots,l+(i-1)\cdot g$ from $s_1$. This is because no other window can cover these characters. For the remaining characters the problem reduces to finding window-compatible transformations for the first $l+(i-2)\cdot g$ characters of $s_1$ and the first $l+(j-1)\cdot g$ characters of $s_2$ with respect to $\langle w_1, w_2, . . ., w_{i-1}\rangle$ and $\langle w'_1, w'_2, . . ., w'_j\rangle$, which is captured by $c_{i-1, j}$, thus $c_{i,j}=c_{i-1,j}+g$. Likewise, we can formulate the case that $w'_j$ is unmatched by $c_{i,j}=c_{i,j-1}+g$.

Note that $c_{k,k'}$ is equivalent to an optimal window-compatible transformation from $s_1$ to $s_2$. By iterating through $i$ from $1$ to $k$ and $j$ from $1$ to $k'$ one can simply calculate $c_{i,j}$ in time $O(1)$ according to \eqref{eq:windowdynamic}. Therefore $c_{k,k'}$ can be calculated in time $O(kk')$, and the proof is complete.
\end{proof}

\begin{corollary}
Given an $\alpha$-approximation matrix of edit distances between the substrings corresponding to every pair of windows of $W_1$ and $W_2$, one can compute an $\alpha$-approximation of the optimal window-compatible transformation of $s_1$ into $s_2$ in time $O(n+|W_1||W_2|)$.
\end{corollary}

\begin{lemma}\label{compatibility}
Given that $\edit(s_1,s_2) \leq \delta n$, there exists a window-compatible transformation of $s_1$ into $s_2$ with respect to $W_1$ and $W_2$ that has at most $(3\delta +  1/\gamma)n + 2l$ operations.
\end{lemma}

\begin{proof}
Recall that $l$ is the length of the windows, and $\gamma$ is the number of layers. Let \opt\enspace be a minimum size transformation of $s_1$ into $s_2$. The overall idea of the proof is as follows. We first show that there exists a set of non-overlapping windows of length $l$, such that a window-compatible transformation with respect to them approximates \opt. Next, by shifting those windows and losing a small fraction on the approximation factor, we fit them to those in $W_1$ and $W_2$.

Consider a pair of characters $x\in s_1$ and $y\in s_2$, such that \opt\enspace transforms $x$ into $y$ either with no change or through a substitution. We call such a pair an edge. Note that there is no collision in the set of all edges in \opt\enspace (or generally in any transformation), i.e. for edges $(x_1, y_1)$ and $(x_2, y_2)$ if $x_1<x_2$ then $y_1<y_2$. Let $M=\langle (x_1,y_1),\ldots,(x_m,y_m)\rangle$ be the sequence of all edges in \opt\enspace in order from left to right.

Now we find the first set of windows as follows. Roughly speaking, we iterate through $M$ and at each step put as many edges as possible in a window of length $l$. In particular, let $\rho(i)$ be the smallest index in $M$ such that $x_{\rho(i)}$ and $y_{\rho(i)}$ are not covered by any window up to step $i$. We create window $v_i$ of length $l$ starting from $x_{\rho(i)}$ and window $v'_i$ of length $l$ starting from $y_{\rho(i)}$. We stop when any such window goes beyond the length of the strings.

In this way, there might be some edges that have one endpoint in $v_i$ and one endpoint beyond $v'_i$ or vice versa. Consider the case in which these edges have one endpoint in $v_i$. Let $h(i)$ be the number of such edges, and let $p(i)$ be the number of characters in $v'_i$ that \opt\enspace transforms them through insertion. We claim that $p(i)\geq h(i)$. This is because $v'_i$ has at most $l-h(i)$ edges in \opt. In comparison, a transformation with respect to $v_i$ and $v'_i$ can keep all the edges between $v_i$ and $v'_i$ and apply deletion and insertion for those edges that have one endpoint in $v_i$ and one endpoint out of $v'_i$. This costs at most $2h(i)$.

\begin{figure}[h!]
	\begin{center}
		\includegraphics[scale=0.7]{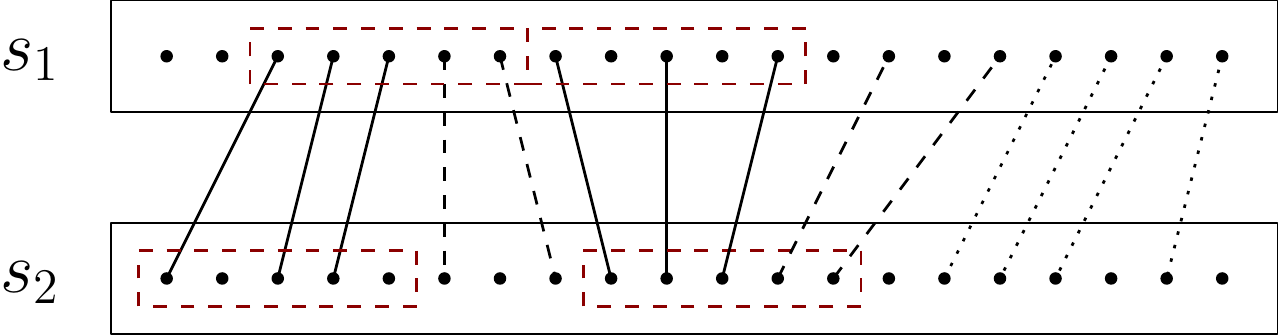}
	\end{center}
	\caption{\small{An edge from $x\in s_1$ to $y\in s_2$ shows that \opt \enspace transforms $x$ into $y$ with either no change or a substitution. Dashed edges represent those with one covered endpoint. Dotted edges represent those that are remained at the end of the iteration.}}
	\label{fig:wincomp1}
\end{figure}

Besides, the number of remaining edges at the end of the iteration on $M$ is at most $l$. Such edges can be transformed by at most $2l$ insertions and deletions. Hence, requiring the transformation to be with respect to all $v_i$'s and $v'_i$'s adds at most $2l+2\sum h(i)\leq 2l+2\sum p(i)\leq 2l+2|\opt|\leq 2l+2\delta n$ more operations to the optimum solution. Equivalently, the optimum transformation with respect to these windows has at most $3\delta n+2l$ operations.

\begin{figure}[h!]
	\begin{center}
		\includegraphics[scale=0.7]{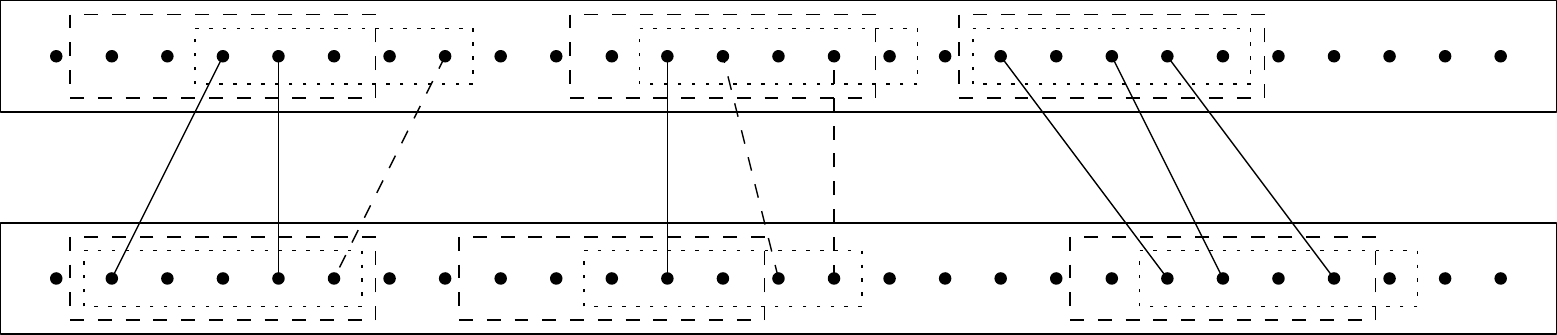}
	\end{center}
	\caption{\small{Dotted rectangles represent $v_i$'s and $v'_i$'s. Dashed rectangles represent shifted windows that are in $W_1$ or $W_2$. Dashed lines represent edges that are left outside of windows after shifting.}}
	\label{fig:wincomp1}
\end{figure}

Finally, we note that the gap size between the windows in $W_1$ is $g=l/\gamma$, therefore, one can shift $v_i$'s by at most $g/2$ to the right or left in order to map them to non-overlapping windows in $W_1$. Likewise, one can find non-overlapping windows for $v'_i$'s in $W_2$. Every shift of a window leaves at most $g/2$ of the edges outside, which costs an extra $g$ operations. Since there are at most $n/l$ windows, the overall cost of shifting the windows is $n/\gamma$. Therefore, there exists a subset of $W_1$ and $W_2$ such that the optimum window-compatible transformation of $s_1$ into $s_2$ with respect to them has at most $2l+3\delta n+n/\gamma=(3\delta+1/\gamma)n+2l$ operations.
\end{proof}

The next lemma proves the approximation factor and time complexity of our $7+\epsilon$ approximation algorithm for the $\delta$-bounded edit distance problem.

\begin{lemma}\label{mainbutnotmain}
There exists a quantum algorithm that solves the $\delta$-bounded edit distance problem within an approximation factor of $7+\epsilon$ in time $\tildorder(n^{2-1/7}\poly(1/\epsilon))$.
\end{lemma}

\begin{proof}
First, without loss of generality we can assume that $\delta > n^{-1/14}$, because otherwise one can use the $O(n+d^2)$ algorithm of \landauincremental~\cite{landau1998incremental} for strings of distance at most $d$ and find the exact edit distance in time $O(n+\delta^2 n^2)=O(n^{2-1/7})$.

We prove that the algorithm discussed in Section \ref{section:ourResults} leads to an approximation factor of $7+\epsilon$ in quantum running time $\tildorder(n^{2-1/7}\poly(1/\epsilon))$. To this end, let us go through the algorithm step by step.

Note that the total number of windows is equal to $O(n/g)$ where $g$ is the gap size. Therefore Step (i) of the algorithm takes time $O(n/g)$. In Step (ii), we use the $3+\epsilon$ approximation algorithm for metric estimation to approximate the distances between the windows. The running time of each oracle invocation is $O(l^2)$ since the length of windows are $l$. Also, there are at most $O(n/g)$ points in this metric estimation instance, therefore due to Theorem \ref{thm:metric1} the total running time of Step (ii) is: $$\tildorder((l^2(n/g)^{5/3}+(n/g)^2)\poly(1/\epsilon))\enspace .$$
Note that $g=l/\gamma$. By assigning $l=n^{1-\beta}$ the overall running time of Step (ii) is:
$$\tildorder((n^{2-\beta/3}\gamma^{5/3} + n^{2\beta}\gamma^2)\poly(1/\epsilon))\enspace .$$

Step (iii) takes time $O(n+|W_1||W_2|)$ due to Lemma \ref{dp}, and thus the running time of this step is $O(n+\gamma^2n^{2\beta})$. Thus, the overall running time of the algorithm up to Step (iii) is $$\tildorder((n^{2-\beta/3}\gamma^{5/3})\poly(1/\epsilon)+n+\gamma^2 n^{2\beta})\enspace .$$
By assigning $\beta = 6/7$, the running time of the algorithm becomes $$\tildorder(n^{2-2/7}(\gamma^{5/3}\poly(1/\epsilon)+\gamma^2))\enspace .$$

Finally, by choosing $\epsilon'=\epsilon/4$ and $\gamma=(\epsilon'\delta)^{-1}$ the overall running time of the algorithm becomes $O(n^{2-1/7}\poly(1/\epsilon))$. What remains is to show that assigning such values for $l$ and $\gamma$ gives the $7+\epsilon$ approximation factor. Due to Lemma \ref{compatibility}, there exists a window-compatible transformation of $s_1$ into $s_2$ with respect to $W_1$ and $W_2$ that has at most $(3\delta+1/\gamma)n+2l$ operations. According to the proof, at most $2\delta n$ of these operations are inside the windows. Therefore, a $(3+\epsilon')$-approximation of the distances between the windows in Step (iii) of the algorithm gives us a transformation with at most $(2\delta n)(3+\epsilon')+(\delta+1/\gamma)n+2l$ operations. This can be simplified as follows:
\begin{align*}
	(2\delta n)(3+\epsilon')+(\delta+1/\gamma)n+2l
	&\leq 6\delta n + 2 \epsilon'\delta n + \delta n + \frac{n}{\gamma}+2l &\\
	&\leq (7\delta +2\epsilon'\delta+ \frac{1}{\gamma}+2n^{-6/7})n &\\
	&\leq (7 +2\epsilon'+ \frac{1}{\delta\gamma}+\frac{2n^{-6/7}}{\delta})\delta n &\\
	&\leq (7 +2\epsilon'+ \epsilon'+2 n^{-11/14})\delta n & \delta > n^{-1/14}\\
	&\leq (7+4\epsilon')\delta n & \text{for every $n > (\frac{2}{\epsilon'})^{14/11}$} \\
	&\leq (7+\epsilon)\delta n \enspace .
\end{align*}

Therefore, the algorithm finds a window-compatible transformation of $s_1$ into $s_2$ with respect to $W_1$ and $W_2$ that is $(7+\epsilon)$-approximation and runs in quantum time $O(n^{2-1/7}\poly(1/\epsilon))$.
\end{proof}

\begin{theorem}\label{main}
There exists a quantum algorithm that solves edit distance within an approximation factor of $7+\epsilon$ in time $\tildorder(n^{2-1/7}\poly(1/\epsilon))$.
\end{theorem}

\begin{proof}
Let \opt\enspace be the edit distance between the two strings. We can check if $\opt=0$ in time $O(n)$. Assume that $\opt\geq 1$. We guess a value $\rho$ for  \opt\enspace by iterating through different multiplicative ranges from $1$ to $n$. Let $\epsilon'=\epsilon/9$. In particular, in every step $i\geq 0$ we guess a range $[\delta n,(1+\epsilon')\delta n)$ for \opt, where $\delta=(1+\epsilon')^i/n$, and run the algorithm of Lemma \ref{mainbutnotmain} with parameters $\epsilon'$ and $(1+\epsilon')\delta$. Note that at each step we can verify whether the output of the algorithm is a valid transformation or not. We get the first valid transformation as soon as \opt\enspace lies within the range of our guess. This valid transformation is of size at most $(7+\epsilon')(1+\epsilon')\delta$ which is no more than $(7+\epsilon)\delta n$. Also, there are at most $\log_{1+\epsilon'}(n)\in \tildorder(1/\epsilon)$ ranges for which we run the algorithm of Lemma \ref{mainbutnotmain}. Hence, the overall time for the search is $\tildorder(n^{2-1/7}\poly(1/\epsilon))$.
\end{proof}

\section{Bootstrapping}\label{bootstraping}
Recall that in Section \ref{bbc}, we described our bootstrap algorithm that uses itself as the oracle of the metric distance algorithm. Here, we compute the time complexity and the approximation factor of the algorithm.

\begin{theorem}\label{thm:bootstrap}
There exists an $\tildorder(n^{2-(5-\sqrt{17})/4+\epsilon}\poly(1/\epsilon))$ time quantum algorithm that approximates edit distance within a factor $\erroredit(\epsilon) = \errorvalue$.
\end{theorem}
\begin{proof}
The algorithm is presented in Section \ref{bbc}. Here we prove the claimed time complexity and approximation factor. Suppose the time complexity of our algorithm for the edit distance problem is $\timeedit(\epsilon)=\tildorder(n^{2-\phi_\epsilon}\poly(1/\epsilon))$ and the time complexity of our algorithm for the bounded edit distance problem is $\tildorder((1/\delta)^2 n^{2-2\phi_\epsilon}\poly(1/\epsilon))$. Notice that the total number of windows is equal to $O(n/g)$ and thus Step (i) of the algorithm can be easily done in time $O(n/g)$.
In Step (ii), we use the $O(1/\epsilon)$ approximation algorithm of metric estimation to approximate the distances of the windows. Moreover, we use our algorithm of edit distance $\alg(2\epsilon)$ recursively for the oracle of metric estimation. Notice that, the length of every window is $l$. Furthermore, the number of points in the metric is equal to the number of windows, namely $O(n/g)$. Note that the running time and query complexity of the $O(1/\epsilon)$ algorithm of metric estimation are $\tildorder(n^2\poly(1/\epsilon))$ and $\tildorder(n^{1.5+\epsilon}\poly(1/\epsilon))$, respectively. Therefore, the total running time of this step is
$$\tildorder((\timeedit(2\epsilon)(l)\cdot (n/g)^{1.5+\epsilon}+(n/g)^2)\poly(1/\epsilon)) = \tildorder((n^{(1-\beta_\epsilon)(2-\phi_{2\epsilon})+1.5\beta_\epsilon+\epsilon \beta_\epsilon}\gamma^{1.5+\epsilon} + n^{2\beta_\epsilon}\gamma^2)\poly(1/\epsilon))$$
since $l = O(n^{1-\beta_\epsilon})$ and $n/g = O(\gamma n /l) = O(\gamma n^{\beta_\epsilon})$.

Step (iii) takes time $O(n+|W_1||W_2|)$ due to Lemma \ref{dp} and thus the running time of this step is $O(\gamma^2n^{2\beta})$. Thus, the overall running time of the algorithm is $\tildorder((n^{(1-\beta)(2-\phi)+1.5\beta+\epsilon \beta}\gamma^{1.5+\epsilon} + n^{2\beta}\gamma^2)\poly(1/\epsilon))$. If we set $\beta_\epsilon=(\sqrt{17}-1)/4+\epsilon$ and $\phi_\epsilon = 1-\beta_\epsilon$, the running time of the algorithm for $\delta$-bounded edit distance would be $\tildorder((1/\delta)^2 n^{2-(5-\sqrt{17})/2+2\epsilon}\poly(1/\epsilon))$. Hence, the running time of the algorithm for the edit distance problem is $\tildorder(n^{2-(5-\sqrt{17})/4+\epsilon}\poly(1/\epsilon)) = O(n^{1.781})$.

To compute the approximation ratio, we should first compute the number of nested levels we use the algorithm in itself. In the $i$th recursion, we use $\alg(2^{i-1}\epsilon)$, while it works better than an $O(n^2)$ algorithm. Thus we have
$$2-(5-\sqrt{17})/4+2^{i-1}\epsilon < 2 \implies 2^{i-1}\epsilon < (5-\sqrt{17})/4 \implies 2^i < c/\epsilon \implies i<\log_2(1/\epsilon)+o(1)$$
Hence, we have at most $\log(1/\epsilon)+o(1)$ levels of recursion. We know that $\erroredit(\epsilon) = 2\errormetric(\epsilon)\erroredit(2\epsilon)+1$, because if we get an $\alpha$ approximate of the optimal window-compatible transformation, it is a $2\alpha +1$ approximate of the optimal solution as discussed in Lemma \ref{mainbutnotmain}. We also know that $\errormetric(\epsilon) = O(1/\epsilon)$. Hence, we can compute the edit distance as follows.
\begin{equation*}
\begin{split}
\erroredit(\epsilon) &= 2\errormetric(\epsilon)\erroredit(2\epsilon)+1
 \leq (c/\epsilon) \erroredit(2\epsilon) = \frac{c^i}{1\cdot 2\cdot . . . \cdot 2^{i-1}\cdot \epsilon^i}\\
 &=\frac{(1/\epsilon)^{c'}}{(1/\epsilon)^{(i-1)/2} \cdot \epsilon^i} = O(1/\epsilon)^{O(\log 1/\epsilon)}\\
\end{split}
\end{equation*}
This completes the proof.
\end{proof}

At last, we discuss why the exponent of our algorithm converges to $2-(5-\sqrt{17})/4$. Recall the recursive formula of for computing the running time of the algorithm:
$$\tildorder(n^{2-\phi_\epsilon}\poly(1/\epsilon)) = \tildorder((n^{(1-\beta_\epsilon)(2-\phi_{2\epsilon})+1.5\beta_\epsilon+\epsilon \beta_\epsilon}\gamma^{1.5+\epsilon} + n^{2\beta_\epsilon}\gamma^2)\poly(1/\epsilon))\enspace.$$
Intuitively, the running is the maximum of two terms, and the best is when these terms are equal. Thus when $\epsilon\to 0$ we can roughly tell:
$2-2\phi_0 = 2\beta_0 = (1-\beta_0)(2-\phi_0)+1.5\beta_0$. This equation has only one positive answer which is $\beta_0 = (\sqrt{17}-1)/4$ and $\phi_0 = 1-\beta_0$; therefore, the exponent is equal to $2-\phi_0 =2-(5-\sqrt{17})/4$.

\section{Approximating Edit Distance in MapReduce}\label{mapreduce}
Edit distance has been studied in parallel and distributed models since the 90s. However, the sequential nature of the dynamic programming solution makes it difficult to parallelize; therefore most of these solutions are slow or require lots of memory/communication. Using our framework, we give a somewhat balanced parallel algorithm for the edit distance problem in MapReduce model. More precisely, we give a ($3+\epsilon$)-approximation algorithm which uses $O(n^{8/9})$ machines, each with a memory of size $O(n^{8/9})$. Moreover, our algorithm runs in a logarithmic number of rounds and has time complexity $O(n^{1.704})$ on one machine which is truly subquadratic.
The overall communication and total memory of our algorithm are also truly subquadratic, due to the sublinearity of the number of machines and the memory of each machine.

Our algorithm is significantly more efficient than previous PRAM algorithms, for instance \cite{apostolico1990efficient} in terms of the number of machines, the overall memory, and the overall communication.
In addition, this is the first result of its kind for edit distance in MapReduce model.
Although this subject has been studied before, previous studies targeted a different aspect of the problem, such as giving a heuristic algorithm, an algorithm for inputs from a particular distribution model, or an algorithm for edit distance between all pairs of several strings \cite{editmapreduce2}.

We begin by stating some of the MapReduce notions and definitions in Section \ref{basics} and next explain our algorithm is Section \ref{mapreducealgorithm}.

\subsection{MapReduce Basics}\label{basics}

In this section, we give a brief overview of the MapReduce setting and later show how our framework can be used to design a MapReduce algorithm for edit distance.

In the MapReduce model, 
an algorithm consists of several rounds. Each round has a mapping phase and a reducing phase. Every unit of information is represented in the form of a $\langle key; value\rangle$ pair in which both key and value are strings. The input, therefore, is a sequence of $\langle key; value\rangle$ pairs specifying the input data and their corresponding positions. For instance, in the case of edit distance, we assume the input pairs are either in the form of $\langle (s_1,i); s_1[i]\rangle$ or $\langle (s_2,i); s_2[i]\rangle$ where the value represents a character, and the key shows the position of this character in either $s_1$ or $s_2$.

Each round of a MapReduce algorithm is performed as follows: every single input pair is given to a mapper separately and depending on the mapping algorithm, a sequence of $\langle key; value\rangle$'s is generated with respect to the input key. Note that the mappers have to be \textit{stateless} in the sense that the output of every mapper is only dependent on the single $\langle key; value\rangle$ pair given to it. Since the mappers are stateless, parallelism in the mapping phase is straightforward; all the inputs are evenly distributed between the machines. Moreover, there is no limit on the types of the $\langle key; value\rangle$ outputs that the mappers generate. Once \textit{all} the mapper jobs are finished, the reducers start to run. Let $\mathcal{K}$ be the set of all keys generated by the mappers in the mapping stage. In the reducing stage, every $key \in \mathcal{K}$ along with all its associated values is given to a single machine. Note that there is no limit on the number of keys generated in the mapping phase as long as all the outputs together fit in the total memory of all machines. However, the values associated with every key should fit in the memory of a single machine since all such values are processed at once by a single reducer. Every reducer, upon receiving a key and a sequence of values associated to it $\langle key; v_1, v_2, v_3, \ldots, v_l\rangle$ runs a reducer-specific algorithm and generates a sequence of output pairs. Unlike the mapping phase, the output keys of a reducer should be identical to the input key given to them. Moreover, the reducers are not stateless since they have access to all values of a key at once, but they can only access their given key and the values associated with it and should be regardless of the other $\langle key; value\rangle$ pairs generated in the mapping phase. Similar to the mapping phase, the total size of the outputs generated by all reducers should no exceed the total memory of all machines together. In addition to this, the total outputs of a reducer should not be more that its memory. Once \textit{all} reducers finished their jobs, the outputs are fed to the mappers for the next round of the algorithm.

For a problem with input length $n$, the goal is to design a MapRuduce algorithm running on $N_p$ machines each having a memory of $N_m$. $N_p$ and $N_m$ have to be sublinear in $n$ since the input is assumed to be huge in this setting. Moreover, since the overhead of a MapReduce round is time-consuming, the number of MapReduce rounds of the algorithms should be small (either constant or polylogarithmic). Many classic computational problems have been studied in the MapReduce setting. For instance, \karloffmodel~\cite{karloff2010model} provide a MapReduce algorithm to compute an MST of a graph with a sublinear number of machines and a sublinear memory for every machine. \lattanzifiltering~\cite{lattanzi2011filtering} design a filtering method and based on that, provide MapReduce algorithms for fundamental graph problems such as maximal matchings, weighted matchings, vertex cover, edge cover, and minimum cuts.

We show in Section \ref{mapreducealgorithm} that using $O(n^{8/9})$ machines and $O(n^{8/9})$ memory on each machine, one can design a MapReduce algorithm for edit distance that runs in $O(\log n)$ MapReduce rounds. Moreover, the running time of the algorithm is subquadratic.

\subsection{Edit Distance in MapReduce}\label{mapreducealgorithm}
Our solution for approximating edit distance in MapReduce uses the same framework explained in Section \ref{editdistance}. Therefore, we solve the problem by solving the $\delta$-bounded edit distance problem several times. The difference is that here we solve all of these subproblems simultaneously. This only imposes a multiplicative factor of $O((1/\epsilon)\log n)$ to the number of machines and a multiplicative factor of $1+\epsilon$ to the approximation factor, hence in the following, we focus on solving the $\delta$-bounded edit distance problem.

We use two different approaches for large $\delta$'s and small $\delta$'s. For large $\delta$'s, we use our framework and compute the edit distance between some pairs of windows of $s_1$ and $s_2$ all at once. For small $\delta$'s though, we use a new method based on $(\mn, +)$ matrix multiplication, also known as distance multiplication. We denote it by $\star$.
We separate the large and the small $\delta$'s with a critical value based on the number of machines\footnote{for $n^{8/9}$ machines $\delta^* = n^{-8/27}$.}.

For $(\mn, +)$ matrix multiplication in the MapReduce model, we use a parameterized version of the algorithm presented in \cite{saeedmapreduce}.

\begin{theorem}[Proved in \cite{saeedmapreduce}]\label{minplusfirst}
For any two $n\times n$ matrices $A$ and $B$ and $0 < x \leq 2$, $A \star B$ can be computed with $n^{3(1-x/2)}$ machines and memory $O(n^x)$ in $1 + \lceil (1-x/2)/x\rceil$ MapReduce rounds. Moreover, the total running time of the algorithm is $O((1/x)n^3)$.
\end{theorem}

Given that we have a chain of matrices to be multiplied instead of just two matrices, we can use Theorem \ref{minplusfirst} to halve the number of matrices in two rounds; therefore we have the following corollary.
\begin{corollary}[of Theorem \ref{minplusfirst}]
\label{corollary:mm}
The $(\mn, +)$ multiplication of $n^a$ matrices of size $n^b\times n^b$ can be computed in $2\lceil a\log_2 n\rceil$ rounds of MapReduce with $n^y$ machines for any $0\leq y\leq a+3b/2$, with a memory of $O(n^{2(a+3b-y)/3})$ for each machine.
Moreover, the running time of the algorithm (for one machine) is $\tildorder(n^{a+3b-y})$.
\end{corollary}

Notice that for two $n\times n$ matrices in Corollary \ref{corollary:mm}, we have $a=0$ and $b=1$, hence the number of machines is $n^y$ and the memory of each machine is $O(n^{2 - 2y/3})$ which is the same as Theorem \ref{minplusfirst} where $x = 2 - 2y/3$. Also note that for $0\leq y\leq a+3b/2$, we use Theorem \ref{minplusfirst} with $1\leq x\leq 2$, hence all $1/x$ terms are ignored.

In Sections \ref{sec:largedelta} and \ref{sec:smalldelta}, we discuss our approach for large $\delta$'s and small $\delta$'s, respectively. In Section \ref{sec:conclusion}, we discuss the remaining details of the algorithm.

\subsubsection{Our Approach for Large $\delta$'s}
\label{sec:largedelta}
The overall idea of our solution for large $\delta$'s is to use our framework as follows: we first construct some windows for each string, then we find the edit distance between some pairs of windows, and afterward we find a window-compatible transformation, which is a good approximation to the desired edit distance between two input strings.

The first step of our approach is to find the edit distance between some pairs of windows. Previously, we found an approximated edit distance between all pairs of windows using metric estimation. On the contrary, here we can do better than finding the edit distance between all pairs based on the following observation. 

\begin{lemma}
\label{lemma:restrict}
Given that $\edit(s_1,s_2) \leq \delta n$, there exists a window-compatible transformation of $s_1$ into $s_2$ with respect to $W_1$ and $W_2$, that for each window $w_i\in W_1$ that matches to a window $w_2\in W_2$, their indices do not differ by more than $\lceil \delta n/g \rceil$, and the number of operations is at most $(3\delta +  1/\gamma)n + 2l$.
\end{lemma}
\begin{proof}
This lemma is similar to Lemma \ref{compatibility} with an additional condition that the indices of any two matching windows do not differ by more than $\lceil \delta n/g \rceil$. The proof is also similar to Lemma \ref{compatibility}.

Let \textsf{opt} be a minimum size transformation of $s_1$ into $s_2$. Consider a pair of characters $x\in s_1$ and $y\in s_2$ such that \textsf{opt} transforms $x$ into $y$ either with no change or through a substitution. As before we call such a pair an edge. Let $M = \langle(x_1,y_1),\dots,(x_m,y_m)\rangle$ be the sequence of all edges in opt in order from left to right.

Similar to the proof of Lemma \ref{compatibility}, we first find a set of non-overlapping windows of length $l$ by iterating over $M$. In particular, let $\rho(i)$ be the smallest index in $M$ such that no window covers $x_{\rho_i}$ and $y_{\rho_i}$ up to step $i$. We create a window $v_i$ of length $l$ starting from $x_{\rho_i}$ and a window $v{i′}$ of length $l$ starting from $y_{\rho_i}$. We stop when any such window goes beyond the range of the strings or all edges of $M$ are covered.
We then shifted these windows to the left to become consistent with the windows of $W_1$ and $W_2$. We proved there exists a window-compatible transformation with respect to those windows with at most $(3\delta +  1/\gamma)n + 2l$ operations. Now we complete the proof using same windows, by showing that the indices of any matching windows do not differ by more than $\lceil \delta n/g \rceil$.

We know that (for example see Corollary $1$ of \cite{ukkonen1985algorithms}) for any edge $(x_{\rho_i}, y_{\rho_i})\in M$, their indices differ by at most $\edit(s_1, s_2)\leq \delta n$. Therefore, for any $v_i$ and $v'_i$, positions of their first characters differ by at most $\delta n$. From this, we can directly conclude that the indices of the corresponding shifted windows in $W_1$ and $W_2$ differ by at most $\lceil \delta n/g \rceil$.
\end{proof}

We find the edit distance between useful pairs of windows in the first round. To do this, we give some pairs of windows to a machine and use the na\"{i}ve DP-based algorithm to find the edit distance between them. In the next round, we combine the results of the first round to find the best window-compatible transformation. The second round is similar to Lemma \ref{dp}; the difference is that the memory and the running time is slightly reduced by Lemma \ref{lemma:restrict}. The second round uses only one machine.

We have the following lemma for large $\delta$'s (or small $\alpha$'s). To simplify the notation, let $\delta = n^{-\alpha}$.

\begin{lemma}\label{largedeltaslemma}
We can solve the $\delta$-bounded edit distance problem for 
\begin{itemize}
\item
$0\leq x\leq 13/20$ and $\alpha\leq 3(x+1)/16$ with $n^x$ machines, and $O((1/\epsilon^2)n^{(11-5x)/8+\epsilon'})$ memory for each machine in time $O((1/\epsilon^2)n^{(35-13x)/16})$ (for one machine), and for
\item
$13/20\leq x \leq 7/6$ and $\alpha\leq 2(4-x)/21$ with $n^x$ machines, and $O((1/\epsilon^2)n^{2(4-x)/7+\epsilon'})$ memory for each machine in time $O((1/\epsilon^2)n^{(50-23x)/21})$ (for one machine).
\end{itemize}
in two MapReduce rounds, where $\epsilon'>0$ is an arbitrary constant.
\end{lemma}
\begin{proof}
We already stated the sketch of our algorithm. To analyze the algorithm, we define and set some parameters carefully.

Recall that we used two parameters of $\beta$ and $\gamma$ to construct the windows of length $l=\lfloor n^{1-\beta}\rfloor$ with a gap size $g=\lfloor l/\gamma \rfloor$ for each of the input strings. Lemma \ref{lemma:restrict} states that given $\edit(s_1,s_2)\leq \delta n = n^{1-\alpha}$, there exists a window-compatible transformation with at most $(3\delta +  1/\gamma)n + 2l$ operations. To keep the approximation factor as small as $3+\epsilon$, we should have $n/\gamma \ll \delta n$ and $2l\ll \delta n$. Setting $\gamma = 1/\delta \epsilon$ and $\beta > \alpha$ suffice.
By doing this, the number of windows for each string is at most $$n_{w_1}, n_{w_2} \leq n\gamma/l = O((1/\epsilon)n^{\alpha+\beta}).$$

In the first round, we find the edit distance between all \textit{useful pairs} of  windows, which are in fact the pairs with an edit distance of at most $\lceil \delta n/g\rceil$. By Lemma \ref{lemma:restrict}, the number of such useful pairs is at most
$$min(|W_1|,|W_2|)\cdot(2\lceil \delta n/g\rceil+1) = O((1/\epsilon^2)n^{\alpha+2\beta}) $$

Therefore, if we have $n^x$ machines, every machine gets $O((1/\epsilon^2)n^{\alpha+2\beta-x})$ pairs.
The edit distance between a pair of windows can be computed in time $O(l^2)$ and memory $O(l)$ where $l=\lfloor n^{1-\beta} \rfloor$. Hence, the memory of each machine in round 1 is $O((1/\epsilon^2)n^{1+\alpha+\beta-x})$. Moreover, the time complexity of each machine in this round is $O((1/\epsilon^2)n^{2+\alpha-x})$.

In the second round, we only use one machine to combine the results of the first round. This machine has to get edit distance between all pairs from all machine; hence it needs $O((1/\epsilon^2)n^{\alpha+2\beta})$ memory. The time complexity of this round is also $O((1/\epsilon^2)n^{\alpha+2\beta})$.

By setting $\beta = \alpha+\epsilon'/2$ for an arbitrary constant $\epsilon'>0$, and setting $\alpha$ as stated in the lemma, we get the desired result.
\end{proof}

\subsubsection{Our Approach for Small $\delta$'s}
\label{sec:smalldelta}
The other side of the edit distance problem is the case when the two given strings are similar. In this case, if we try to use our framework, we would encounter too many windows, and this exceeds the time and memory given to the algorithm. Previously, in this case, we used the algorithm of \landauincremental{}~\cite{landau1998incremental} with time $O(n+d^2)$. This solution cannot (trivially) become parallel. Here, we instead use a novel approach based on $(\mn, +)$ matrix multiplication. We again use the fact that a character $c_1$ from $s_1$ can only be transformed (with no change or a substitution) to a character $c_2$ in $s_2$ only if their positions differ by at most $\edit(s_1, s_2)$~(Corollary $1$ of \cite{ukkonen1985algorithms}).

Let $d(i,j+1,i',j'+1)$ be the edit distance between two substrings of $s_1[i,\dots,j]$ and $s_2[i',\dots,j']$. We have the following lemma.
\begin{lemma}\label{smalldeltaslemma}
\label{lem:dij}
For an arbitrary $k$, $i<k\leq j$, we have:
\begin{equation*}
d(i,j+1,i',j'+1) = \substack{min\\i'-1\leq k'\leq j'}\big\{d(i,k+1,i',k'+1) +d(k+1,j+1,k'+1,j+1)\big\}.
\end{equation*}
\end{lemma}
\begin{proof}
We construct a transformation from $s_1[i,\dots,j]$ to $s_2[i',\dots,j']$ using two transformation: one from $s_1[i,\dots,k]$ to $s_2[i',\dots,k']$ and the other from $s_1[k+1,\dots,j]$ to $s_2[k'+1,\dots,j']$, therefore $$d(i,j+1,i',j'+1) \leq \substack{min\\i'-1\leq k'\leq j'}\big\{d(i,k+1,i',k'+1)+d(k+1,j+1,k'+1,j+1)\big\}.$$
Also, if we define $k^*$ as the largest index that a character from $s_1[i,\dots, k]$ transfers into $s_2[k^*]$ in an optimal transformation, or $k^*=i'-1$ if no such index exists, we have $$d(i,j+1,i',j'+1) = d(i,k+1,i',k^*+1)+d(k+1,j+1,k^*+1,j+1).$$
which completes the proof.
\end{proof}

Moreover, computing $d(i,j+1,i',j'+1)$ is useful only when $|i-i'|\leq d$ and $|j-j'|\leq d$~(Corollary $1$ of \cite{ukkonen1985algorithms}), therefore for a fixed $i$ and $j$, all of these \textit{useful values} form a $(2\delta n+1) \times (2\delta n+1)$ matrix, namely $D^{i,j}$. Rewriting Lemma \ref{lem:dij} in matrices, we have the following corollary.
\begin{corollary}[of Lemma \ref{lem:dij}]
\label{cor:matdij}
For an arbitrary $k$, $i\leq k\leq j$, we have $D^{i,j}=D^{i,k}\star D^{k,j}$, where $\star$ is the $(\mn, +)$ matrix multiplication operator.
\end{corollary}

Notice that $\edit(s_1,s_2)=d(1, |s_1|+1, 1, |s_2|+1)$, which is an element of $D^{1,|s_1|}$. To compute this matrix, we do as follows:
for a parameter $y$, $0\leq y \leq 1$, which we'll fix later, we partition $s_1$ into $n^y$ substrings of length at most $n^{1-y}$. Each of these substrings has a matching substring in $s_2$ with a length at most $n^{1-y}+2\delta n$. Using the na\"{i}ve DP-based algorithm, we construct a $(2\delta n+1) \times (2\delta n+1)$ matrix for each of these $n^y$ substrings in the first round. The matrices are $D^{1, t}, D^{t+1, 2t}, \dots, D^{(\lceil|s_1|/t\rceil-1)t+1, |s1|}$ where $t=n^{1-y}$. By Corollary \ref{cor:matdij} we have $D^{1,|s_1|} = D^{1, t}\star D^{t+1, 2t} \star \dots \star D^{(\lceil|s_1|/t\rceil-1)t+1, |s1|}$. Therefore, we obtain the result in remaining rounds by the matrix multiplication algorithm of Corollary \ref{corollary:mm}.

\begin{lemma}
We can solve the $\delta$-bounded edit distance problem for
\begin{itemize}
\item
$0\leq x\leq 13/20$ and $\alpha\geq 3(x+1)/16$ with $n^x$ machines, and $O(n^{(11-5x)/8})$ memory of each machine in time $O(n^{(51-29x)/16})$ (for one machine), and for
\item
$13/20\leq x \leq 7/6$ and $\alpha\geq 2(4-x)/21$ with $n^x$ machines, and $O(n^{2(4-x)/7})$ memory of each machine in time $O(n^{(58-25x)/21})$ (for one machine).
\end{itemize}
in at most $O(\log n)$ MapReduce rounds.
\end{lemma}
\begin{proof}
Here, we analyze the described algorithm in more details.
In the first round, constructing the full matrix is a time-consuming process for one machine; therefore we break this job into $n^t$ parts. More precisely, we partition rows of the solution matrix into $n^t$ parts and give the task of computing each part to one machine. Therefore, in the first round, the number of machines is equal to $n^{y+t}=n^x$. The memory of each machine is the maximum of its input size, its running memory, and its output size, which are equal to $2n^{1-y}+2d$, $O(n^{1-y})$, and $n^{2-2\alpha-t}$, respectively. The time complexity of one machine using the DP-based algorithm is $O(n^{1-y}\cdot n^{1-y} \cdot n^{1-\alpha-t}) = O(n^{3-2y-\alpha-t})$.

The second part of the algorithm is analogous to Corollary \ref{corollary:mm} where $a=y$ and $b=1-\alpha$, therefore if the number of machines is $n^x$, the memory of each machine is $n^{2(y+3(1-\alpha)-x)/3}$. 
Moreover, the time complexity of one machine is $\tildorder(n^{3-3\alpha+y-x})$.

Setting $y=(6\alpha+2x-3)/5$ and $t=x-y$ give us the desired result. Also, note that the range of $x$ is consistent with Corollary \ref{corollary:mm}.
\end{proof}

\subsubsection{Conclusion}
\label{sec:conclusion}
We compute edit distance by solving the $\delta$-bounded edit distance problems for several $\delta$'s in parallel. For each $\delta=n^{-\alpha}$ we use the appropriate MapReduce algorithm based on the value of $x$ and $\alpha$. When all subproblems are finished, we also have a final round for combining the results of these subproblems to obtain the final (approximated) edit distance. Therefore, the desired MapReduce ($3+\epsilon$)-approximation algorithm for edit distance is as follows.

\begin{theorem}
We can solve the edit distance problems in MapReduce model in at most $O(\log n)$ MapReduce rounds with $\tildorder((1/\epsilon)n^x)$ machines and for
\begin{itemize}
\item
$0\leq x\leq 13/20$ with a memory of at most $O((1/\epsilon^2)n^{(11-5x)/8+\epsilon'})$ for one machine in time $O(n^{(51-29x)/16})$ (for one machine), and for
\item
$13/20\leq x \leq 7/6$ with a memory of at most $O((1/\epsilon^2)n^{2(4-x)/7+\epsilon'})$ in time $O(n^{(58-25x)/21})$ (for one machine).
\end{itemize}
\end{theorem}
\begin{proof}
We solve the problem for $\delta=0$ and $\delta=(1+\epsilon/3)^k/n$ for $0\leq k\leq O((1/\epsilon)\log n)$ in parallel machines. For each subproblem, we use $\lceil n^x\rceil$ machines.

In the zero case, we only check whether $s_1=s_2$ or not. This can be done with at most $n^{1-x}$ memory and  $O(n^{1-x})$ for each machine. We handle other subproblems by Lemmas \ref{largedeltaslemma} or \ref{smalldeltaslemma}\footnote{in fact, for small $\delta$'s we can run our algorithm just once for the largest $\delta$}. Therefore, the memory of each machine is at most $O((1/\epsilon^2)n^{(11-5x)/8+\epsilon'})$ for $0\leq x\leq 13/20$ and $O((1/\epsilon^2)n^{2(4-x)/7+\epsilon'})$ for $13/20\leq x \leq 7/6$.

The time complexity of each machine is the maximum time of Lemmas \ref{largedeltaslemma} and \ref{smalldeltaslemma} which is at most $O(n^{(51-29x)/16})$ for $0\leq x\leq 13/20$ and $O(n^{(58-25x)/21})$ for $13/20\leq x \leq 7/6$. The number of rounds is also at most $O(\log n)$.
\end{proof}

By setting $x=8/9$, we minimize the maximum of the number of machines and the memory of each machine. This is shown in Figure \ref{mapreducetradeoff}.

\begin{figure*}
\begin{center}
\begin{tikzpicture}[
  declare function={
    func(\x)= (\x<=0.65) * (11-5*\x)/8   +
     (\x>0.65) * 2*(4-\x)/7;
     ftime(\x)= (\x<=0.65) * (51-29*\x)/16   +
     (\x>0.65) * (58-25*\x)/21;    
  }
]

\begin{axis}[domain=0:1.1667, ymin=0.2,ymax=1.4, title={The trade-off between the number of machine and memory of each machine},     xlabel = $x$ (exponent of the number of machines),
ylabel = {exponent of memory},
xtick={0,0.2,0.4,0.6,0.8,1.0,1.2},
ytick={0,0.2,0.4,0.6,0.8,1.0,1.2,1.4,1.6,1.8,2.0},
ymajorgrids=true, xmajorgrids=true,
grid style=dashed]
\addplot[color={rgb:green,1;black,3;white,2}]{func(x)};
\addplot[dashed, domain=0:1.1667]{x};

\addplot[
color=blue,
mark=square*,
]
coordinates {
	(0.889,0.889)
};

\end{axis}
\end{tikzpicture} 
\end{center}
\caption{The trade-off between the number of machines and memory of each machine is shown. In $x = 8/9$ the maximum of the number of machines and the memory of each machine is minimized.}
\label{mapreducetradeoff}\end{figure*}
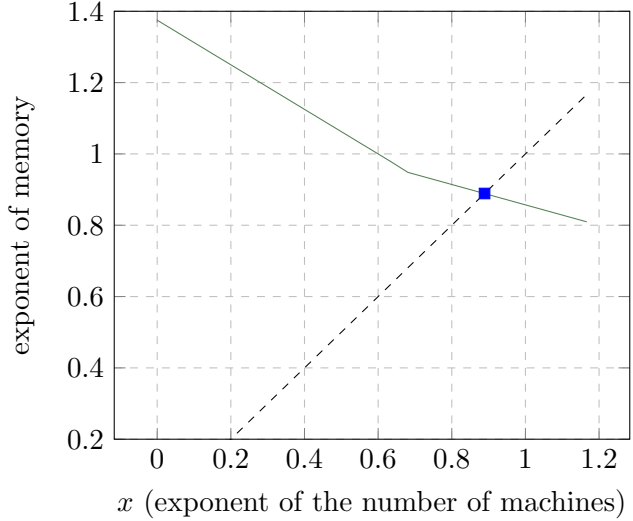

\begin{corollary}
We can solve the edit distance problems in MapReduce model with an approximation factor of $3+\epsilon$ in $O(\log n)$ rounds with $\tildorder((1/\epsilon)n^{8/9})$ machines, a memory of $O((1/\epsilon^2)n^{8/9+\epsilon'})$ for each machine, and in time $O(n^{2-8/27})$ (for one machine), where $\epsilon'>0$ is an arbitrary constant.
\end{corollary}
\section{Other Similarity Measures} \label{section:otherMeasures}
Edit distance is one of many similarity measures for comparing two strings. Furthermore, it is one of many problems with a simple two-dimensional DP solution. Other measures and similar problems include longest common subsequence (\textsf{lcs}), Fr\'{e}chet distance (\textsf{fre}) and dynamic time warping (\textsf{dtw}). While the $O(n^2)$ solution for these problems are very analogous, unfortunately, our approach does not directly apply to them. In the following, we discuss some reasons behind this difficulty. The update rule of these measures are defined as follows:
\begin{align*}
\edit(i,j)=\,&min\big\{\edit(i-1,j)+1,\edit(i,j-1)+1,
 \edit(i-1,j-1)+(s_1[i]\neq s_2[j])\big\}\\
\lcs(i,j)=\,&max\big\{\lcs(i-1,j),\lcs(i,j-1), 
\lcs(i-1,j-1)+(s_1[i]\neq s_2[j])\big\}\\
\dtw(i,j)=\,&min\big\{\dtw(i-1,j), \dtw(i,j-1),
 \dtw(i-1,j-1)\big\}+dis(i,j)\\
\fre(i,j)=\,&max\big\{min\{\fre(i-1,j),\fre(i,j-1),
\fre(i-1,j-1)\}, dis(i,j)\big\}
\end{align*}

Our framework for approximating edit distance is based on two assumptions. First, the usability of Lemma \ref{compatibility}, which states that there is a window-compatible solution which is a good approximation to the optimal solution. Second, to use the metric estimation, the desired measure should be a distance function, namely a metric. 

Two similarity measures \textsf{dtw} and \textsf{lcs} are not metric, moreover they cannot be approximated by any metric. For example, for \textsf{dtw} consider $s_1=a^{2k+1}$, $s_2=a^kba^k$ and $s_3=ab^{2k-1}a$. We have $\dtw(s_1, s_2)=1$ and $\dtw(s_2, s_3)=0$, but $\dtw(s_1, s_3)=2k-1$. Therefore the triangle inequality does not hold here.

The similarity measure \textsf{lcs} is in fact, the opposite of a metric function, i.e., for two similar strings, their \textsf{lcs} is large, and for two different strings, their \textsf{lcs} is small. The first property of a distance function does not hold here, for a non-empty string s, $\lcs(s,s)\neq 0$.
The other part of our approach where \textsf{lcs} has a drawback is the Lemma \ref{compatibility}. For a window size $l$, one can consider $s_1 = (ab^{l-1}a^{l-1})^t$ and $s_2=(a^lc^{l-1})t$. We have $\lcs(s_1, s_2)=lt$, but $\lcs$ of a windows-compatible transformation is at most $t$.

Likewise, approximating $\lcs$ in classic computers is also harder that $\edit$. None of the results for approximating $\edit$ is shown for $\lcs$, unless when $lcs(s_1, s_2)=\Omega(n)$.
Another way around this is to approximate \textsf{co-lcs} instead of $\lcs$, where \textsf{co-lcs}$(s_1, s_2)=|s_1|+|s_2|-\lcs(s_1, s_2)$. This measure is very similar to edit distance but without the substitution operation. Using our framework, we can approximate \textsf{co-lcs} with the same approximation factor of $7+\epsilon$ in quantum computers and an approximation factor of $3+\epsilon$ in MapReduce.

Fr\'{e}chet distance is rather a similarity measure for curves instead of strings. For strings, the problem becomes trivial, i.e., zero for same strings and one for different strings. However, \textsf{fre} on curves has a similar dynamic programming solution to \textsf{edit}. This similarity in solution leads us to consider this problem, too.
If we study the problem regardless of its geometric properties, i.e. all distances are given as a matrix, we can prove that approximating \textsf{fre} is as hard as computing its exact value.

\begin{theorem}
\label{thm:fre}
If there exists a quantum (or MapReduce) approximation algorithm for Fr\'{e}chet distance with a constant approximation factor in time $O(n^{2-\epsilon})$, which takes distances as a matrix in the input, there also exists a quantum or MapReduce algorithm which computes the exact Fr\'{e}chet distance in time $O(n^{2-\epsilon})$.
\end{theorem}

\begin{proof}
The idea is a gap producing reduction from the problem to itself. We can do a binary search on the actual value of Fr\'{e}chet distance, thus in each step for a threshold $t$, we want to know whether $\fre(a,b)\leq t$ or not. We define a new distance function as:
\begin{equation*}
dis'(a,b) =
\begin{cases}
1 & dis(a,b)\leq t, \\
n^2 & \text{otherwise.}
\end{cases}
\end{equation*}
If $\fre(a,b)\leq t$ then we have $\fre'(a,b)=1$ and $\fre'(a,b)=n^2$ otherwise. Therefore, if we solve the new instance with an approximation algorithm with a constant factor, we can decide whether $\fre'(a,b)=1$ or not. Thus we can decide whether $\fre(a,b)\leq t$ or not. Hence, we can find the exact value of Fr\'{e}chet distance by executing the approximation algorithm in $O(\log n)$ round. This argument works for both quantum algorithms and MapReduce algorithms.
\end{proof}

Theorem \ref{thm:fre} does not rule out the possibility of a subquadratic quantum algorithm or MapReduce algorithm for Fr\'{e}chet distance, but it states that relaxing the problem in this way does not make the problem easier.

\section{Conclusion and Open Problems}
In the quantum algorithm of Section \ref{editdistance}, we have a fixed length for windows, $l=n^{1/7}$. By allowing windows of different sizes, we can improve the approximation factor to $3+\epsilon$.
Rubinstein, Schramm, and Song independently improved this factor to $3+\epsilon$~\cite{lcdquantum}.
Moreover, by redefining the metric estimation problem in a way that we only output a diagonal band of the distance matrix, meaning the main diagonal and zero or more diagonals on either side, we can improve the running time of our main algorithm to $\tildorder(n^{38/21}\poly(1/\epsilon)) = \tildorder(n^{1.810})$.
The same technique improves the running time of our bootstrapping algorithm of Section \ref{bootstraping} to $O(n^{1.708})$.

Indeed the most important open problem concerning edit distance is whether a subquadratic time algorithm can approximate the edit distance of two strings within a constant factor?
In this regard, our paper proposes the following approach. Suppose we want to approximate the pairwise edit distance between $m$ given strings, each of size $n$, within a constant factor. We call this problem \textit{pairwise edit distance}. A na\"{i}ve solution for pairwise edit distance has running time $O(m^2n^2)$. Obviously, any subquadratic time algorithm for approximating edit distance within a constant factor improves upon this running time. 
In this paper, we show that an improvement in the running time of pairwise edit distance also leads to a subquadratic time algorithm for approximating edit distance within a constant factor.
We believe this actually opens a new direction for approximating edit distance for classic computers.

In addition to this, our work gives rise to a number of questions that we believe are important to study in future work.
\begin{itemize}
	\item \textit{How efficiently can we approximate metric estimation in classic computers with a subquadratic number of queries, when the distance function is edit distance?}
	\item \textit{Is there subquadratic quantum algorithms that approximate other similarity measures, \textsf{LCS} in particular?}	
	\item \textit{Can a quantum algorithm approximate the edit distance of two strings within a constant factor in near-linear time?}
	\item \textit{Is it possible to show a non-trivial lower bound on the quantum computational complexity of computing edit distance?}
\end{itemize}
\section{Acknowledgment}
We would like to thank Andrew Childs, Omid Etesami, Salman Beigi, and Mohammad Ali Abam for their comments on an earlier version of the paper.

\bibliographystyle{abbrv}
	
\bibliography{edit}

\end{document}